\documentclass[letterpaper,twocolumn]{quantumarticle} 
\pdfoutput=1
\makeatletter
\providecommand{\@afterenddocumenthook}{}
\makeatother
\usepackage[utf8]{inputenc}
\usepackage[english]{babel}
\usepackage[T1]{fontenc}
\usepackage{amsmath}
\usepackage{amsthm}
\usepackage{hyperref}

\usepackage{tikz}
\usepackage{lipsum}
\usepackage{amssymb}
\usepackage{physics}
\usepackage{comment}
\usepackage[normalem]{ulem}

\newtheorem{lemma}{Lemma}[section]
\newtheorem{theorem}{Theorem}[section]
\newtheorem{corollary}{Corollary}[theorem]

\usepackage[ruled,vlined,linesnumbered]{algorithm2e}

\definecolor{crimson}{RGB}{186,0,44}
\definecolor{moss}{RGB}{0, 186, 111}

\begin{document}

\title{Parallel quantum signal processing via polynomial factorization}

\author{John M. Martyn} 
\affiliation{Center for Theoretical Physics, Massachusetts Institute of Technology, Cambridge, MA 02139, USA}
\orcid{0000-0002-4065-6974}

\author{Zane M. Rossi} 
\affiliation{Department of Physics, Massachusetts Institute of Technology, Cambridge, MA 02139, USA}
\orcid{0000-0002-7718-654X}

\author{Kevin Z. Cheng}
\affiliation{Department of Physics, Massachusetts Institute of Technology, Cambridge, MA 02139, USA}

\author{Yuan Liu}
\affiliation{Department of Electrical and Computer Engineering, North Carolina State University, Raleigh, NC 27606, USA}
\affiliation{Department of Computer Science, North Carolina State University, Raleigh, NC 27606, USA}
\affiliation{Department of Physics, North Carolina State University, Raleigh, NC 27606, USA}
\orcid{0000-0003-1468-942X}

\author{Isaac L. Chuang}
\affiliation{Department of Physics, Massachusetts Institute of Technology, Cambridge, MA 02139, USA}
\affiliation{Department of Electrical Engineering and Computer Science,
Massachusetts Institute of Technology, Cambridge, MA 02139, USA}
\orcid{0000-0001-7296-523X}

\maketitle

\begin{abstract}
  Quantum signal processing (QSP) is a methodology for constructing polynomial transformations of a linear operator encoded in a unitary. Applied to an encoding of a state $\rho$, QSP enables the evaluation of nonlinear functions of the form $\tr(P(\rho))$ for a polynomial $P(x)$, which encompasses relevant properties like entropies and fidelity. However, QSP is a sequential algorithm: implementing a degree-$d$ polynomial necessitates $d$ queries to the encoding, equating to a \emph{query depth} $d$. 
  
  Here, we reduce the depth of these property estimation algorithms by developing \emph{Parallel Quantum Signal Processing}. Our algorithm parallelizes the computation of $\tr (P(\rho))$ over $k$ systems and reduces the query depth to $d/k$, thus enabling a family of time-space tradeoffs for QSP. This furnishes a property estimation algorithm suitable for distributed quantum computers, and is realized at the expense of increasing the number of measurements by a factor $O( \text{poly}(d) 2^{O(k)} )$. We achieve this result by factorizing $P(x)$ into a product of $k$ smaller polynomials of degree $O(d/k)$, which are each implemented in parallel with QSP, and subsequently multiplied together with a swap test to reconstruct $P(x)$. We characterize the achievable class of polynomials by appealing to the fundamental theorem of algebra, and demonstrate application to canonical problems including entropy estimation and partition function evaluation.
\end{abstract}

\section{Introduction}


\noindent The increasing sophistication of quantum computers has pressured researchers to clearly demarcate problems for which quantum algorithms provide provable speedups over their classical counterparts. 
Toward resolving this tension, recent work has proposed a qualified `unification of quantum algorithms', based on the related frameworks of quantum signal processing (QSP) \cite{ylc_fixed_point_14, Low_2016, Low_2017, Low_2019} and the quantum singular value transformation (QSVT) \cite{Gily_n_2019}. These algorithms enable the application of tunable polynomial functions to the singular values of large linear operators, in turn unifying and simplifying the presentation of most known quantum algorithms \cite{martyn2021grand}, while simultaneously exhibiting good numerical properties \cite{Haah_2019, chao2020finding, Dong_2021, dlnw_infinite_qsp_22}, near-optimal query complexity \cite{ms_query_comp_fun_24}, and fruitful connections to well-studied techniques in matrix decomposition and functional analysis \cite{cs_qsvt_tang_tian_23,dlnw_infinite_qsp_22,almtw_szego_qsp_24, berntson2024complementary}.

Despite this broad success, we are not yet in a world of fault tolerant quantum devices and are thus unable to leverage the full apparatus of QSP/QSVT. Experimental implementations of QSP/QSVT on existing hardware have largely been limited to small examples on noisy devices \cite{debry_qsp_23, kikuchi2023realization}. Consequently it is an interesting question whether the theoretical success of QSP and QSVT can be extended to experimental systems with mild resource constraints, such as limited coherence times. Specifically, QSP and QSVT are sequential algorithms without intermediate measurement, and indeed derive their pleasing properties from their circuit depths. Nonetheless, it seems reasonable that even with a limited circuit depth some of these properties could be recovered at the cost of increased circuit width and/or sample complexity.

As suggested by the title of this work, we are interested in realizing this by parallelizing QSP, motivated by the ubiquitous use of parallel processing in classical computation \cite{ps_architecture_11}. At a high-level, such processes take advantage of the fact that certain problems permit division into simpler sub-problems. Showing that such divide-and-conquer strategies can prove advantageous requires that (1) the original problem can be subdivided efficiently, (2) the sub-problems can be solved faster than the original problem, and (3) solutions to sub-problems can be efficiently reconciled to produce the full solution. We translate this notion to the circuit model of quantum computation in a straightforward way, where sub-processes (analogously to `threads' in the classical world) are considered in parallel when they act on disjoint subsets of qubits, and where problem subdivision takes place entirely \emph{classically}, before the execution of the quantum circuit. 

This work proposes one scheme for parallelizing QSP problems into multiple independent sub-problems, each of which requires shallower circuits. This technique, termed \emph{Parallel Quantum Signal Processing} (Parallel QSP) is applicable to the task of computing \emph{nonlinear functions} of quantum states, with wide application in the estimation of common properties like entropies and entanglement measures \cite{Johri_2017,Elben_2018,Buhrman_2001}, and the efficient realization of disparate multi-state tests \cite{oszmaniec_relational_24}. As QSP generates polynomial transformations, subdivision of a problem will correspond to polynomial factorization, and reconciliation to multiplication of factor polynomials.

Our construction sits at the intersection of two lines of work for estimating properties of quantum states: quantum signal processing \cite{low2017quantum,Gily_n_2019}, and multivariate trace estimation \cite{Johri_2017, Quek_2024}. 
The benefit of combining these techniques is rooted in jointly leveraging their individual strengths. While QSP can prepare nonlinear functions of a quantum state, the circuits required are often quite deep~\cite{Gily_n_2019,tan2023error,rycs_noisy_qsp_22}. Conversely, while recently proposed methods for multivariate trace estimation can make do with shallow circuits (e.g., constant depth~\cite{Quek_2024}), the achievable class of nonlinear functions is limited. 
This work provides candidate problems for which these two toolkits can be applied jointly, with the benefit of analytic simplicity and asymptotic savings in circuit depth.

\subsection{Results and Paper Outline}

In standard QSP, generation of a degree-$d$ polynomial requires $d$ successive queries to the encoding unitary, corresponding to a \emph{query depth} $d$ and circuit depth $O(d)$. Given that large circuit depths are prohibitive on near-term devices limited by short coherence times, we seek valid methods to parallelize QSP into many shorter QSP circuits and reduce the corresponding query depth.

Here we parallelize computation over $k$ threads and reduce query depth by an $O(k)$ factor. In practical situations, $k = O(1) \ll d$, such as when parallelizing a large degree polynomial over a few quantum devices. This parallelization establishes parallel QSP as a suitable algorithm for distributed quantum computation~\cite{caleffi2022distributed}, where multiple devices can be run concurrently. However, the reduction in depth afforded by parallelization comes with an increased number of measurements $O(\text{poly}(d)2^{O(k)})$. Crucially, this scales polynomially in $d$ rather than exponentially, in contrast to techniques like error mitigation that feature super-polynomial scaling in depth~\cite{quek2022exponentially}. On the other hand, this comes with an exponential cost in the number of threads, similar to that seen in quantum circuit cutting~\cite{Peng_2020}. 

An informal statement of our main result is given in the following Theorem~\ref{thm:informal_main}. For intuition, we also provide a comparison of standard QSP and parallel QSP in Fig.~\ref{fig:generic_overview}.
\begin{figure}
    \includegraphics[width=0.95\linewidth]{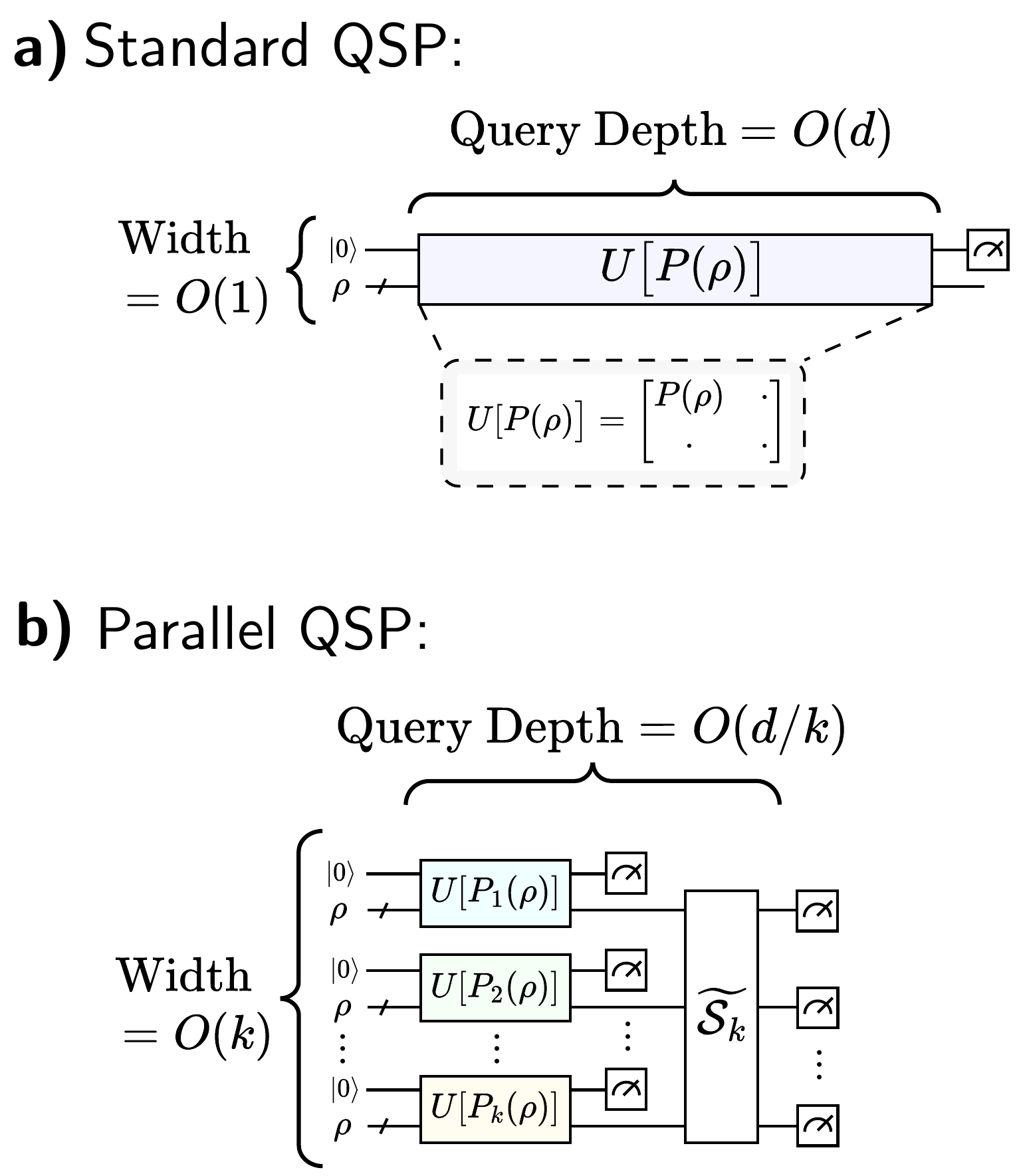}
    \caption{Illustration of standard QSP (\textbf{a}) vs.~parallel QSP (\textbf{b}). The operators $U[P(\rho)]$ denote block-encodings of $P(\rho)$ (as depicted in the inset), realized through QSP. For a degree-$d$ polynomial, standard QSP generally requires query depth $2d = O(d)$. In contrast, parallel QSP distributes the computation over $k$ threads by implementing factor polynomials $P_k(x)$ in parallel, and achieves a reduced query depth $O(d/k)$. This is accomplished using a \emph{swap test}, schematically denoted here by an operation $\widetilde{\mathcal{S}_k}$ and subsequent measurements; see Sec.~\ref{sec:Generalized_Swap_Test} for explanation. }
    \label{fig:generic_overview}
\end{figure}

\begin{theorem}[Informal statement of Theorem~\ref{thm:parallel_qsp_prop_est_2}] \label{thm:informal_main}
    Let $P(x)$ be a real-valued polynomial of degree $d$, that is bounded as $\max_{x\in[-1,1]} | P(x) | \leq 1$. Given access to an input state $\rho$ and a block encoding thereof, we can invoke parallel QSP across $k$ threads to estimate the property
    \begin{equation}
        w = \tr(P(\rho))
    \end{equation}
    with a circuit of width $O(k)$ and query depth at most $\approx d/2k $. The number of measurements required to resolve $w$ to additive error $\epsilon$ is 
    \begin{equation}
    \begin{aligned}
        &O\Bigg( \frac{ {\rm poly}(d) 2^{O(k)}}{\epsilon^2} \Bigg),
    \end{aligned}
    \end{equation}
    where the terms ${\rm poly}(d) 2^{O(k)}$ depend on the chosen factorization of $P(x)$.     
\end{theorem}

We provide an explicit polynomial decomposition and factorization that ensures the measurement overhead scales as $O\big( \|P\| \cdot d^4 (1+\sqrt{2})^{4k}/k^2 \big)$ for a function norm $\| P \|$ defined later. This norm depends on the polynomial $P$, and in general scales as a polynomial in $d$, and in some cases even remains constant. In the worst possible case (corresponding to the most ill-behaved polynomial), we show that the function norm scales as $O(d^{2k+1}/ (k!)^2)$, i.e., as a degree-$2k+1$ polynomial in $d$, that decays super-factorially with $k$ and dominates the $2^{O(k)}$ factor in the overhead. 

Toward an exposition of the main theorem, we review QSP and its application to density matrices and trace estimation in Sec.~\ref{sec:Prelims}. Following this, in Sec.~\ref{sec:parallel_qsp} we present parallel QSP, including a characterization of the achievable class of polynomials. We then adapt parallel QSP to arbitrary property estimation problems in Sec.~\ref{sec:PQSP_for_Property_Estimation}, and exemplify this construction in Sec.~\ref{sec:applications} for the estimation of R\'enyi entropies, partition functions, and the von Neumann entropy. Discussion and comparison with alternative methods are included in Sec.~\ref{sec:discussion}, with detailed proofs of results confined to the appendices.

\section{Preliminaries}\label{sec:Prelims}
In this section, we review the preliminaries for parallel QSP: standard QSP (Sec.~\ref{sec:Overview_QSP}), its application to density matrices (Sec.~\ref{sec:QSP_Density_Matrices}), and its use in estimating the trace of matrix functions (Sec.~\ref{sec:QSP_Trace_Est}).

\subsection{Quantum Signal Processing}\label{sec:Overview_QSP}
Quantum signal processing (QSP) is a method for realizing a polynomial transformation of a quantum subsystem~\cite{Low_2016, Low_2017, Low_2019}. The QSP algorithm works by interleaving a \textit{signal operator} $U$, and a \textit{signal processing operator} $S$. Conventionally, $U$ is taken to be an $x$-rotation through a fixed angle and $S$ a $z$-rotation through variable angle $\phi$:
\begin{equation}
    U(x) = \begin{bmatrix}
        x & i\sqrt{1-x^2} \\
        i\sqrt{1-x^2} & x
    \end{bmatrix}, \qquad 
    S(\phi) = e^{i\phi Z}.
\end{equation}
Introducing a set of $d+1$ \emph{QSP phases} $\vec{\phi} = (\phi_0, \phi_1, ... , \phi_d) \in \mathbb{R}^{d+1}$, the following \textit{QSP sequence} is defined as an interleaved product of $U$ and $S$:
\begin{equation}\label{eq:QSP_seqeunce}
    \begin{gathered}
        U_{\vec{\phi}}(x) = S(\phi_0) \prod_{i=1}^d  U(x)  S(\phi_i).
    \end{gathered}
\end{equation}
The matrix elements of the QSP sequence are manifestly polynomials of $x$:
\begin{equation}\label{eq:qsp}
    U_{\vec{\phi}} = \begin{bmatrix}
        P(x) & iQ(x)\sqrt{1-x^2} \\
        iQ^*(x)\sqrt{1-x^2} & P^*(x)
    \end{bmatrix},         
\end{equation}
where $P(x)$ and $Q(x)$ are polynomials parameterized by $\vec{\phi}$ that obey
\begin{equation}\label{eq:qsp_conditions}
    \begin{split}
        & \text{1. } {\rm deg}(P) \leq d, \ {\rm deg}(Q) \leq d-1 \,; \\
        & \text{2. } P(x)\ \text{has parity } d \bmod 2, \text{ and } Q(x)\ \text{has parity } \\
        & \ \ \ (d-1)\bmod 2 \,; \\
        & \text{3. } |P(x)|^2 + (1-x^2) |Q(x)|^2 = 1, \ \forall \ x \in [-1,1] \, . 
    \end{split}
\end{equation}
This result implies that one can construct polynomials in $x$ by projecting into a block of $U_{\vec{\phi}}$, e.g. $\langle 0| U_{\vec{\phi}} | 0\rangle = P(x)$. Importantly, realizing a degree-$d$ polynomial necessitates $d$ sequential calls to the signal operator, translating to a \textit{query depth} $d$, or a circuit depth $O(d)$. 

While the conditions of Eq.~\eqref{eq:qsp_conditions} restrict the class of realizable polynomials, a broader class of polynomials can be implemented by projecting into other bases and using extensions of QSP. For instance, the $|+\rangle \langle +|$ matrix element $\langle+|U_{\vec{\phi}}|+\rangle = \text{Re}(P(x)) + i\text{Re}(Q(x))\sqrt{1-x^2}$ can realize any real polynomial of definite parity, obviating condition 3 above. Even more powerful is \textit{generalized QSP}, introduced in Ref.~\cite{motlagh2023generalized} as an extension of the QSP sequence in Eq.~\eqref{eq:QSP_seqeunce}. As we review in Appendix~\ref{app:Arbitrary_polynomials}, generalized QSP enables one to design an arbitrary polynomial $P(x)$ restricted only by the condition $|P(x)| \leq 1$ over $x \in [-1,1]$. This encompasses complex polynomials and those of indefinite parity. To realize a degree-$d$ polynomial, generalized QSP requires a query depth $2d$.

By this reasoning, QSP can encode polynomials that need only be bounded as $\|P\|_{[-1,1]} \leq 1$, where $\| \cdot \|_{[-1,1]}$ is the function norm 
\begin{equation}\label{eq:function_norm}
    \| f \|_{[-1,1]} := \max_{x\in[-1,1]} |f(x) | .
\end{equation}
For an arbitrary degree-$d$ polynomial, the requisite query depth is $2d$; however, the query depth reduces to $d$ for a polynomial of definite parity. 
In addition, the converse of this result holds: for any polynomial $\| P \|_{[-1,1]} \leq 1$, there exist corresponding QSP phases that can be efficiently computed with a classical algorithm~\cite{Haah_2019, chao2020finding, dong2021efficient, Ying_2022, yamamoto2024robust}, thus amounting to a classical pre-computation step.

Remarkably, the methodology of QSP can be extended to prepare a polynomial transformation of a Hermitian operator through its extension to the quantum eigenvalue transformation (QET)~\cite{Low_2017, Low_2019, Gily_n_2019}. This is achieved analogous to QSP: provided access to a unitary that block-encodes an operator $A$ in a matrix element, we can design a sequence that encodes a polynomial transformation $P(A)$:
\begin{equation}\label{eq:QET_QSVT_seq}
    U[A]=\begin{bmatrix}
        A & \cdot \\
        \cdot & \cdot
    \end{bmatrix} \ \mapsto \ U_{\vec{\phi}}[A] = 
    \begin{bmatrix}
        P(A) & \cdot \\
        \cdot & \cdot
    \end{bmatrix} ,
\end{equation}
where the unspecified entries ensure unitarity. Unitarity also requires $\| A\|\leq 1$ and $\|P(A)\|\leq 1$; otherwise, these entries must be rescaled by a constant to meet these conditions. Mirroring Eq.~\eqref{eq:QSP_seqeunce}, $U_{\vec{\phi}}[A] $ is an interleaved sequence of $U[A]$ and parameterized rotations. Essentially, this applies QSP within each eigenspace of $A$ and outputs a degree-$d$ polynomial transformation $P(A)$. As above, the cost of realizing an arbitrary degree-$d$ polynomial is $2d$ sequential queries to the block-encoding of $A$, translating to a query depth $2d$, although this reduces to $d$ for a polynomial of definite parity. Lastly, while Eq.~\eqref{eq:QET_QSVT_seq} specializes to a block-encoding in the $|0\rangle\langle 0|$ matrix element, one can more generally take $A$ to be accessed by orthogonal projectors $\Pi, \Pi'$ as $A=\Pi\, U[A]\, \Pi '$.

\subsection{QSP On Density Matrices}\label{sec:QSP_Density_Matrices}
To exemplify QSP, let us consider its application to a density matrix $\rho$. This requires a block encoding of $\rho$, which is directly achievable (sans rescaling) because the norm $\| \rho  \| \leq 1$ for any state. 

While in principle there exist various methods to block encode a density matrix $\rho$, a sufficient oracle is a unitary that prepares a purification of $\rho$~\cite{Low_2019}. This oracle model is known as the \emph{quantum purified query access model}, and has been used in recent works on quantum entropy estimation and property testing~\cite{Subramanian_2021, gilyen2019distributional}. To see how this model works, let $\rho = \sum_j p_j |\chi_j \rangle \langle \chi_j | $ be an $n$-qubit density matrix, and $V_\rho$ be a unitary that prepares a purification of $\rho$ as 
\begin{equation}
    V_\rho | 0\rangle^{\otimes 2n} = |\psi_\rho \rangle_{AB} = \sum_j \sqrt{p_j} |j\rangle_A |\chi_j\rangle_B 
\end{equation}
on $n$-qubit subsystems $A$ and $B$, such that $\tr_A (| \psi_\rho \rangle \langle \psi_\rho |_{AB} ) = \rho_B$. Then, introduce an additional $n$-qubit system $C$, and let $\text{SWAP}_{BC}$ be an operator that swaps subsystems $B$ and $C$. One can then show that $\rho$ is block encoded in the operator $U[\rho]:= (V_\rho^\dag)_{AB} \cdot \text{SWAP}_{BC} \cdot (V_\rho)_{AB} $ as~\cite{Low_2019}
\begin{equation}
   \big( \langle 0 |^{\otimes 2 n}_{AB}  \otimes I_C \big) \cdot U[\rho] \cdot \big( | 0\rangle^{\otimes 2 n}_{AB} \otimes I_C \big) = \rho_C. 
\end{equation}
Therefore, access to $V_\rho$ enables one to block encode $\rho$, and thus apply QSP to generate polynomials $P(\rho)$.

\subsection{Trace Estimation with QSP}\label{sec:QSP_Trace_Est}
Among the applications of QSP, a notable use case is to estimate the trace of a matrix function, i.e. $\tr(f(A))$. There are two methods that we will explore here for estimating such a trace with QSP, both of which approximate $f(A)$ with a QSP polynomial $P(A)$. However, the first method estimates $\tr(P(A))$ with a \textit{Hadamard test}~\cite{aharonov2006polynomial}, while the second method, which we will refer to as the \textit{QSP test}, estimates the squared trace $\tr(|P(A)|^2)$ by measuring the block-encoding qubit(s). Here we discuss both methods and specialize to the block-encoding convention $\Pi = |0\rangle \langle 0 |$ for ease of presentation.

\begin{figure}
    \centering
    \includegraphics[width=0.98\linewidth]{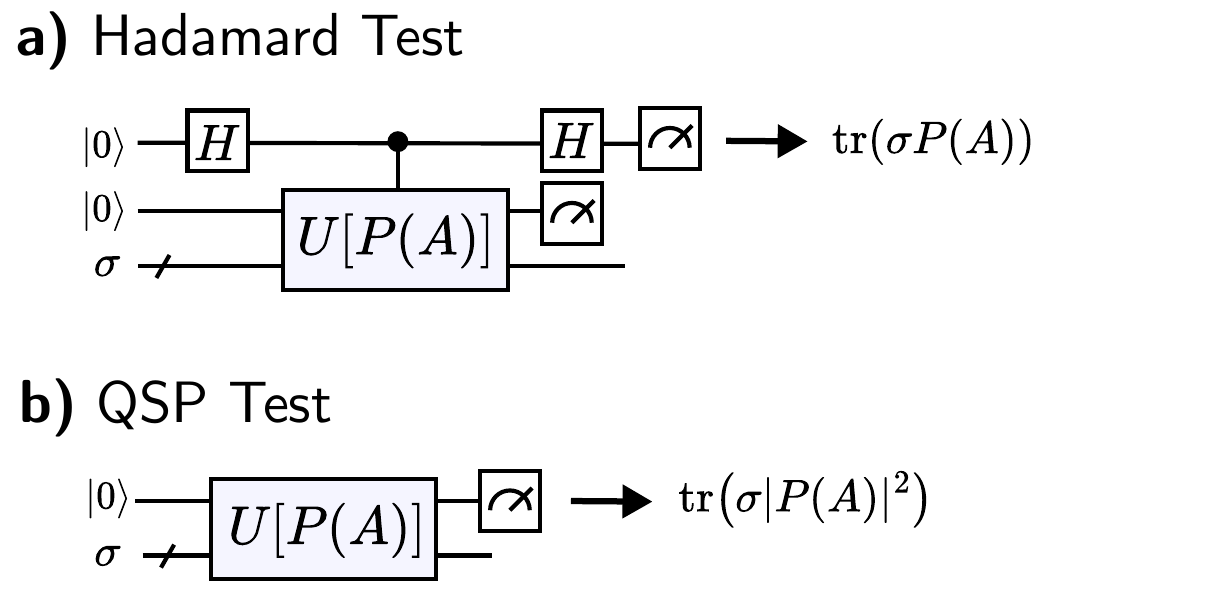}
    \caption{\textbf{a)}: The Hadamard test for trace estimation of QSP polynomials, specialized to a block-encoding in the $|0\rangle \langle 0|$ block. Here, the second register corresponds to the qubit that block-encodes $P(A)$; it is measured and post-selected on the outcome $|0\rangle$, thereby projecting out $P(A)$ and enabling the extraction of the trace $\tr( \sigma P(A) )$. 
    \textbf{b)}: The QSP test for trace estimation of QSP polynomials, also specialized to a block-encoding in the $|0\rangle \langle 0|$ block. Here again, the top register block-encodes $P(A)$ and is measured to extract the trace $\text{tr}( \sigma | P(A) |^2 )$. }
    \label{fig:circuits}
\end{figure}

\subsubsection{Hadamard Test}
In the first method, the Hadamard test~\cite{aharonov2006polynomial} is applied to an input state $\sigma$, with the target unitary set to a QSP sequence that block encodes $P(A)$. We illustrate this circuit in Fig.~\ref{fig:circuits}a. As explained in Lemma 9 of Ref.~\cite{gilyen2022improved}, upon post-selecting the block-encoding qubit(s) to project out the target polynomial $P(A)$, the probability of measuring the ancilla qubit in the state $|0\rangle$ is
\begin{equation}
    p_1 = \frac{1}{2} + \frac{1}{2}\text{Re}\Big[ \text{tr}\big( \sigma P(A) \big) \Big].
\end{equation}
By instead applying a conjugated phase gate to the ancilla qubit after the Hadamard gate, the probability of measuring $|0\rangle$ becomes
\begin{equation}
    p_2 = \frac{1}{2} + \frac{1}{2}\text{Im}\Big[ \text{tr}\big( \sigma P(A) \big) \Big],
\end{equation}
such that the full trace can be reconstructed as $\text{tr}\big( \sigma P(A) \big) = 2p_1-1 + i(2p_2-1)$. By estimating $p_1$ and $p_2$ each to error at most $\epsilon/4$, one obtains an approximation to the trace $\text{tr}\big( \sigma P(A) \big)$ with error at most $\epsilon$. By the central limit theorem, this requires $O(1/\epsilon^2)$ measurements.

\subsubsection{QSP Test}
In the second method, one applies the QSP sequence to an input state $|0\rangle\langle 0| \otimes \sigma$, and then estimates the probability that the block-encoding qubit is measured in the state $|0\rangle$. This is equivalently the probability that the correct block of the QSP sequence is applied to the input state, which is
\begin{equation}
    p = \text{tr}\big(\sigma |P(A)|^2 \big).
\end{equation}
Therefore, an estimation of this to error $\epsilon$ furnishes an approximation of the trace $\text{tr}\big(\sigma |P(A)|^2 \big)$ with error $\epsilon$, and requires $O(1/\epsilon^2)$ measurements. We will refer to this method as the \textit{QSP test}, and depict its circuit in Fig.~\ref{fig:circuits}b.

The QSP test is distinct from the Hadamard test in that the QSP polynomial $P(A)$ is squared in the trace. While the Hadamard test can evaluate traces involving the QSP polynomial $P(A)$ directly (i.e., $\tr(\sigma P(A))$), the QSP test is limited to traces involving its magnitude squared $|P(A)|^2$ (i.e., $\tr(\sigma |P(A)|^2)$). This leads to a tradeoff in the capabilities of these approaches, where the polynomial $|P(A)|^2$ is restricted to be real and non-negative, yet is of twice the degree of $P(A)$.

\subsubsection{Utility in Trace Estimation}
Both the Hadamard test and the QSP test can be used to estimate the trace of a matrix function $\tr(f(A))$ by setting the input state to the maximally mixed state $\sigma=I/2^n$, where $n$ is the number of qubits. By then selecting a polynomial that approximates $f(x)$ as $P(x) \approx f(x)$ or $|P(x)|^2 \approx f(x)$, respectively, both methods output an approximation to $\tr(f(A))/2^n$. Due to this rescaling by $2^n$, resolving $\tr(f(A))$ to error $\epsilon$ requires a number of measurements $O(2^{2n}/\epsilon^2)$. That this cost scales exponentially with the number of qubits is a generic feature, because arbitrary traces $\tr(f(A))$ can be exponentially large in the dimension of $A$. Ref.~\cite{Subramanian_2021} uses this approach to develop an algorithm for estimating the $\alpha$-R\'enyi entropy of a density matrix, and notes the same cost scaling. 

However, for estimating a trace $\tr(f(\rho))$ of a density matrix $\rho$, it is advantageous to set both the input state and block encoding to be $\rho$, i.e., $\sigma = \rho$ and $A = \rho$. In this case, the Hadamard test and QSP test output the traces $\tr(\rho P(\rho))$ and $\tr(\rho |P(\rho)|^2 )$, respectively. From these traces, one can estimate $\tr(f(\rho))$ by selecting polynomials that satisfy $xP(x)\approx f(x)$ or $x|P(x)|^2 \approx f(x)$, respectively. Importantly, this approach circumvents the rescaling by $2^n$, such that estimating $\tr(f(\rho))$ to error $\epsilon$ requires $O(1/\epsilon^2)$ measurements. This streamlined approach is used in Ref.~\cite{wang2022new} to design new algorithms for estimating the $\alpha$-R\'enyi entropy and von Neumann entropy, while avoiding an exponentially large number of measurements. 

Lastly, while both the Hadamard and QSP tests achieve complexity $O(1/\epsilon^2)$, we note that one could alternatively use techniques like amplitude/phase estimation to reduce the complexity to $O(1/\epsilon)$. However, this complexity corresponds to a large $O(1/\epsilon)$ depth~\cite{nielsen2010quantum, Brassard_2002}. As the central focus of this work is reducing depth, we forgo these techniques in favor of the Hadamard and QSP tests, which achieve shallower depths at the expense of an increased measurement overhead.

\section{Parallel Quantum Signal Processing} \label{sec:parallel_qsp}
With the preliminaries laid out, we now present our algorithm for parallel QSP. In its simplest incarnation, parallel QSP enables the estimation of a trace of the form $\tr(\rho^k R(\rho))$, where $R(x)$ is a degree-$d$ polynomial, and $k$ is the number of systems over which the computation is parallelized, i.e., the \textit{number of threads}. While standard QSP can compute this trace with query depth $\sim d+k$, parallel QSP achieves this computation with a query depth $\approx d/k$. This is achieved by factorizing $R(x)$ into $k$ polynomials of degree $O(d/k)$, which are implemented in parallel with QSP, and subsequently multiplied together with a \emph{generalized swap test} (Sec.~\ref{sec:Generalized_Swap_Test}). However, this depth reduction of parallel QSP is realized at the expense of increasing the circuit width to $O(k)$, and the number of measurements by a factor that depends on the chosen factorization of $R(x)$. Moreover, while $\tr(\rho^k R(\rho))$ encompasses a limited class of properties, later in Sec.~\ref{sec:PQSP_for_Property_Estimation} we expand this class to arbitrary polynomial functions.

In this section, we first review the generalized swap test (Sec.~\ref{sec:Generalized_Swap_Test}), which will underpin parallel QSP. We then present the parallel QSP algorithm (Sec.~\ref{sec:PQSP_alg}), including a characterization of the achievable polynomials and a discussion of its resource requirements, and conclude by commenting on the implications of our algorithm (Sec.~\ref{sec:Remarks}).

\subsection{The Generalized Swap Test}\label{sec:Generalized_Swap_Test}
An essential ingredient of parallel QSP is a tool that we will refer to as \textit{the generalized swap test}. As its name suggests this is an extension of the usual swap test, introduced in Ref.~\cite{Buhrman_2001} to measure the overlap between two quantum states, i.e. $\tr(\rho \sigma)$. Explicitly, the generalized swap test uses the identity that the expectation value of a cyclic shift applied to a product state $ \rho^{\otimes k}$ (for an integer $k\geq 1$), is equal to the trace of the multiplicative product~\cite{Ekert_2002}:
\begin{equation}
    \tr( \mathcal{S}_k \cdot \rho^{\otimes k} ) = \tr( \rho^k ),
\end{equation}
where $\mathcal{S}_k$ is a cyclic shift on the $k$ systems comprising $\rho^{\otimes k}$, and acts as 
\begin{equation}
    \mathcal{S}_k \big[ |\psi_1\rangle |\psi_2\rangle |\psi_3\rangle ... |\psi_k\rangle \big] = |\psi_k\rangle |\psi_1\rangle |\psi_2\rangle ... |\psi_{k-1} \rangle.
\end{equation}
Notably, this identity converts a tensor product $\rho^{\otimes k}$ to a multiplicative product $\rho^k$, and reduces to the usual swap test for $k=2$. In addition, this identity holds for a tensor product of distinct states $\rho_j$:
\begin{equation}
    \tr \Bigg( \mathcal{S}_k \cdot \bigotimes_{j=1}^k \rho_j \Bigg) = \tr \Bigg( \prod_{j=1}^k \rho_j \Bigg) . 
\end{equation}

Using the generalized swap test, one can estimate the trace $\tr( \rho^k)$ by measuring the expectation value of $\mathcal{S}_k$ on $k$ copies of $\rho$. Because the states comprising $\rho^{\otimes k}$ can be arranged in parallel in a quantum circuit, this effectively parallelizes the computation of the multiplicative product $\rho^k$, without ever having to explicitly multiply $\rho$ sequentially. Accordingly, the generalized swap test has been employed to compute R\'enyi entropies in quantum Monte Carlo~\cite{Hastings_2010}, estimate nonlinear functions of state on a quantum computer~\cite{Ekert_2002, brun2004measuring, Quek_2024}, and perform entanglement spectroscopy~\cite{Horodecki_2002, Johri_2017, Suba__2019}.

In practice, the expectation value of $\mathcal{S}_k$ can be estimated with various techniques. While an elementary implementation as a Hadamard test applied to the cyclic shift operator translates to a depth $O(k)$~\cite{Johri_2017, Yirka_2021}, recent works have put forth novel constructions of the generalized swap test that achieve $O(1)$ quantum depth~\cite{Suba__2019, Quek_2024}. Ref.~\cite{Suba__2019} achieves this using $2k$ copies of a purification of $\rho$ and additional classical post-processing, leading to an $O(1)$ depth independent of both $n$ and $k$. Alternatively, Ref.~\cite{Quek_2024} prepares an ancilla system in a special GHZ state, from which the cyclic shift $\mathcal{S}_k$ can be measured in depth $O(1)$. Ultimately, these results demonstrate that the generalized swap test can estimate the trace $\tr(\rho^k)$ with a circuit of width $O(k)$ and depth $O(1)$, thus fully parallelizing the computation of the multiplicative product.

\subsection{Parallel QSP}\label{sec:PQSP_alg}
Parallel QSP is a synthesis of the QSP test (Sec.~\ref{sec:QSP_Trace_Est}) and the generalized swap test (Sec.~\ref{sec:Generalized_Swap_Test}). At a high level, parallel QSP works by first using QSP to implement block encodings of $k$ polynomials $ \{ P_j(\rho) \}_{j=1}^k$ across $k$ threads, and separately applying each to an input state $\rho$. Then applying the generalized swap test to the resulting state, we can extract the trace of the corresponding multiplicative product:
\begin{equation}\label{eq:parallel_QSP_trace}
    z := \tr \Bigg(\prod_{j=1}^k P_j(\rho) \rho P_j(\rho)^\dag \Bigg) = \tr \Bigg( \rho^k \prod_{j=1}^k |P_j(\rho)|^2 \Bigg) . 
\end{equation}
By appealing to the fundamental theorem of algebra, the product $\prod_{j=1}^k |P_j(\rho)|^2$ can represent an arbitrary real, non-negative polynomial. If this target polynomial is of degree $d$, then each polynomial factor $P_j(\rho)$ can be guaranteed to have degree at most $\approx \frac{d}{2k}$, thus dividing the query depth by $O(k)$. While the trace $z$ encompasses only a limited class of functions, in Sec.~\ref{sec:PQSP_for_Property_Estimation} we show that an arbitrary polynomial can be decomposed into this form and made amenable to parallel QSP, enabling general property estimation algorithms at reduced query depth.

As a hybrid of the QSP test and the generalized swap test, parallel QSP requires access to both $\rho$ and a block encoding of $\rho$. As shown in Sec.~\ref{sec:QSP_Density_Matrices}, the purified query access model provides an oracle that prepares a purification of $\rho$ and thus furnishes a block encoding of $\rho$. This oracle also provides access to $\rho$ by disregarding the ancilla system, and thus is sufficient for parallel QSP. Nonetheless, this is not the only possibility, as other oracles can also provide access to both $\rho$ and a block encoding thereof.

\subsubsection{The Parallel QSP Algorithm}\label{sssec:parallel_QSP_alg}

To sharpen our analysis, we first present the parallel QSP circuit in Fig.~\ref{fig:Parallel_QSP_Circuit}. The initial state of the algorithm is a product state $\rho^{\otimes k}$ across the $k$ threads, as well as ancilla qubits used to access block encodings. The circuit then consists of (1) the QSP stage, and (2) the generalized swap test stage. The QSP stage comprises $k$ unitaries $\{ U\big[ P_j(\rho) \big]\}_{j=1}^k$ that block encode polynomials $P_j(\rho)$, realized by QSP. We apply each unitary to the input state in parallel and post-select on the successful application of $P_j(\rho)$. Collectively, this succeeds with probability 
\begin{equation}
\begin{aligned}
    \text{Pr}(\text{QSP Success}) &= \prod_{j=1}^k \tr [P_j(\rho) \rho P_j(\rho)^\dag] ,
\end{aligned}
\end{equation}
and outputs the product state
\begin{equation}
    \bigotimes_{j=1}^k \frac{P_j(\rho) \rho P_j(\rho)^\dag}{\tr \big[ P_j(\rho) \rho P_j(\rho)^\dag \big]} .
\end{equation}

Next, we apply the generalized swap test to this product state to compute the trace of the corresponding multiplicative product, which we denote by $\tilde{z}$:
\begin{equation}
\begin{aligned}
    \tilde{z} :&= \tr( \prod_{j=1}^k \frac{P_j(\rho) \rho P_j(\rho)^\dag}{\tr \big[  P_j(\rho) \rho P_j(\rho)^\dag \big]} ) \\
    &= \frac{ \tr( \rho^k \prod_{j=1}^k |P_j(\rho)|^2 ) }{\prod_{j=1}^k \tr \big[ \rho |P_j(\rho)|^2 \big]} = \frac{z}{\text{Pr}(\text{QSP Success})}. 
\end{aligned}
\end{equation}
We can then estimate $z$ by resolving $\tilde{z}$ and $\text{Pr}(\text{QSP Success})$ to sufficient accuracy. This is the essence of parallel QSP: the computation of a trace of a product of polynomials, by executing these polynomials in parallel rather than sequentially.

\begin{figure}[htbp]
    \centering
    \includegraphics[width=0.99\linewidth]{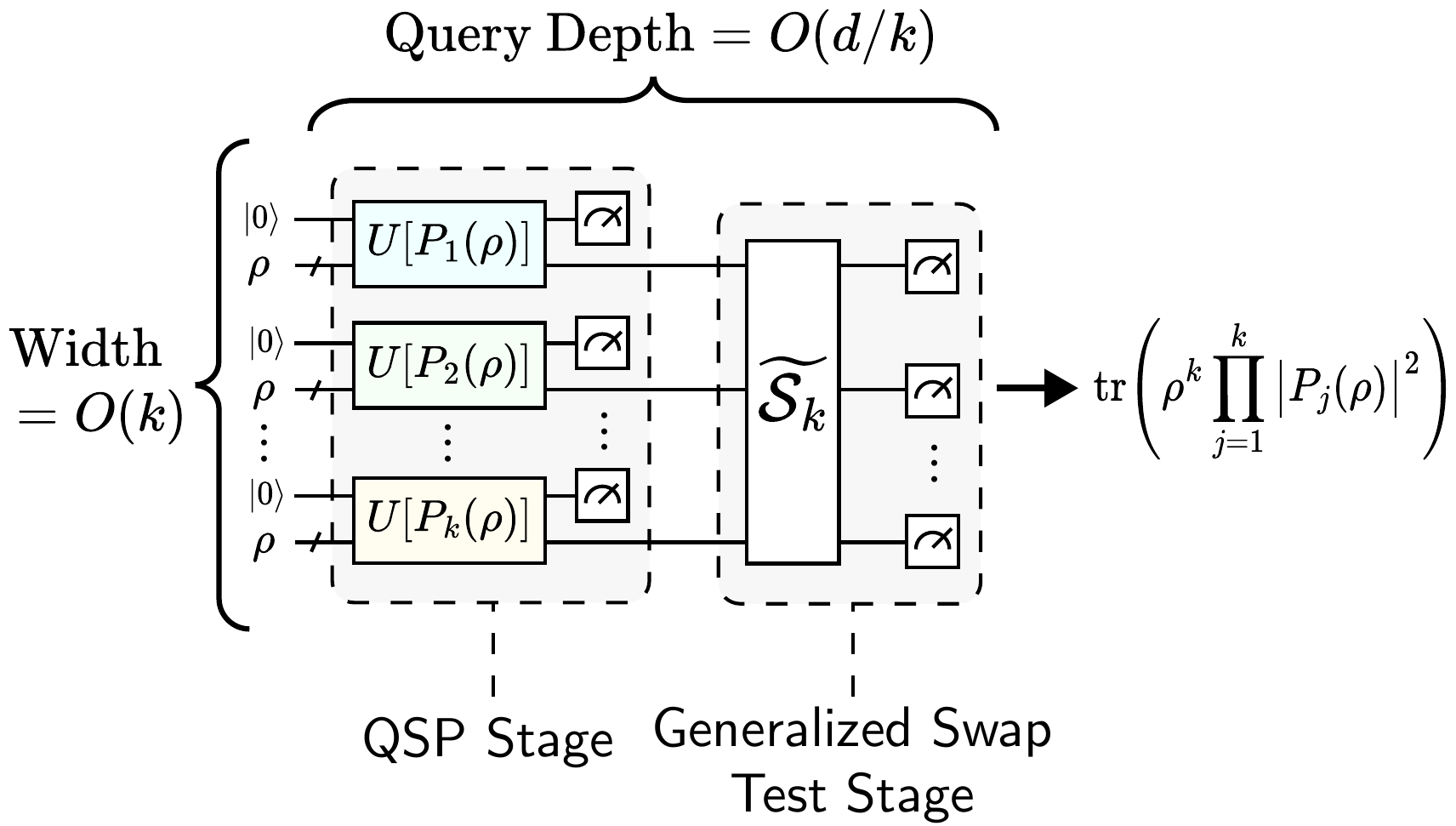}
    \caption{The quantum circuit for the parallel QSP algorithm. The operations $U[P_j(\rho)]$ are unitary block-encodings of polynomials $P_j(\rho)$, realized with QSP. For illustrative simplicity, we specialize to block encodings in the $|0\rangle\langle 0|$ block. Upon application of these polynomials to initial states $\rho$, one enacts a generalized swap test, denoted schematically as $\widetilde{\mathcal{S}_k}$ and subsequent measurements; this serves as a symbolic proxy for the various implementations of the generalized swap test~\cite{Johri_2017, Quek_2024, Suba__2019}. Observe that the parallel QSP circuit has a reduced query depth $O(d/k)$, at the expense of increasing its width to $O(k)$.}
    \label{fig:Parallel_QSP_Circuit}
\end{figure}

With this understanding, we formalize the parallel QSP algorithm with the following theorem: 
\begin{theorem}[Parallel QSP]\label{thm:Parallel_QSP} 
Provided access to a density matrix $\rho$ and a block encoding thereof, the parallel QSP circuit executed across $k$ threads enables the estimation of the quantity 
\begin{equation} \label{parallel_qsp}
    z = \tr \Bigg( \rho^k \prod_{j=1}^k |P_j(\rho)|^2 \Bigg),
\end{equation}
where each $P_j(\rho)$ is a block-encoded polynomial implemented with QSP. More specifically, $z$ can be estimated to additive error $\epsilon$ by running the parallel QSP circuit $O(\frac{1}{\epsilon^2})$ times, where the requisite query depth is $2 \max_j \{ \deg(P_j) \}$ and the circuit width is $O(k)$.
\end{theorem}

\begin{proof}
    Using the parallel QSP circuit of Fig.~\ref{fig:Parallel_QSP_Circuit}, consider the measurement cost of resolving $z = \text{Pr}(\text{QSP Success}) \times \tilde{z}$ to additive error $\epsilon$. Obviously, $\text{Pr}(\text{QSP Success})$ is a probability and can naturally be estimated by repeatedly running the circuit. On the other hand, the expression for $\tilde{z}$ depends on the chosen implementation of the generalized swap test, but in general can be expressed as an expectation value. 
    
    For instance, a simple implementation of the generalized swap test is provided by applying a Hadamard test to $\mathcal{S}_k$. In this case, the probability of measuring the ancilla qubit (of the Hadamard test) in the state $|0\rangle$ upon successful application of each QSP sequence is 
    \begin{equation}
        \text{Pr}\big(\text{Ancilla } = |0\rangle \ \big| \ \text{QSP Success} \big) = \frac{1}{2}(1+\tilde{z}).
    \end{equation}
    Then we can express $z$ as 
    \begin{equation}\label{eq:z_expression_Hadamard_Test}
    \begin{aligned}
        z &= \text{Pr}(\text{QSP Success}) \\
        & \qquad \quad \times \Big( 2 \cdot \text{Pr}\big(\text{Ancilla } = |0\rangle \ \big| \ \text{QSP Success} \big) - 1 \Big) \\
        &= 2 \cdot \text{Pr}\big(\text{QSP Success, and Ancilla } = |0\rangle \big) \\
        & \qquad \qquad - \text{Pr}(\text{QSP Success}).
    \end{aligned}
    \end{equation}
    Therefore, estimating both of these probabilities to additive error $\epsilon/3$ provides an approximation to $z$ with additive error $\epsilon$, and the central limit theorem implies that this requires $O(1/\epsilon^2)$ runs of the parallel QSP circuit. While this specializes to a specific implementation of the generalized swap test, this result is in fact general, because other implementations also approximate $\tilde{z}$ as a combination of expectation values. 
    
    Next, consider the query depth of the parallel QSP circuit. In the QSP stage, each unitary $U\big[ P_j(\rho) \big]$ requires query depth at most $2 \deg(P_j)$ for an arbitrary polynomial as we discussed in Sec.~\ref{sec:Overview_QSP}, corresponding to a total query depth $2 \max_j (\deg(P_j))$. On the other hand, the generalized swap test makes no queries to the block encoding and does not contribute to the query depth. It however can contribute to the circuit depth depending on its implementation, but this can be reduced to $O(1)$ using the constructions of Refs.~\cite{Suba__2019, Quek_2024}
    
    Lastly, as the QSP stage consists of $k$ unitaries enacted in parallel across $k$ systems, its width is $O(k)$. Likewise, while the precise width of the generalized swap test stage depends on its implementation, the constructions of Refs.~\cite{Johri_2017, Yirka_2021, Suba__2019, Quek_2024} use $O(k)$ system copies arranged in parallel, equating to a width $O(k)$.  
\end{proof}

\subsubsection{Characterization of Parallel QSP Polynomials}
According to Theorem~\ref{thm:Parallel_QSP}, the parallel QSP algorithm estimates the trace $z = \tr( \rho^k \prod_{j=1}^k |P_j(\rho)|^2 )$, thus parallelizing the computation of polynomials that take the form $x^k \prod_{j=1}^k |P_j(x)|^2$. We can characterize this class of polynomials by appealing to the fundamental theorem of algebra: 

\begin{lemma}[Factorization of Real, Non-negative Polynomials]\label{thm:Real_NonNeg_Poly_Char}
Consider a polynomial $R(x)$ of even degree $d$ that is real and non-negative over the real axis $x \in \mathbb{R}$. For real inputs $x$, this polynomial can be expressed as the product of the squared magnitudes of $k \leq d$ factor polynomials $\mathcal{R}_j(x)$:
\begin{equation}
    R(x) = \prod_{j=1}^k |\mathcal{R}_j(x)|^2.
\end{equation}
While the factor polynomials $\mathcal{R}_j(x)$ are not unique, there exists a factorization in which every factor polynomial is guaranteed to have degree at most $ \text{deg}(\mathcal{R}_j) \leq \lceil d/2k \rceil$.
\end{lemma}
\begin{proof}
    Applying the fundamental theorem of algebra to a polynomial $R(x)$ of even degree\footnote{Note that the degree is necessarily even; otherwise, the condition of non-negativity for all $x \in \mathbb{R}$ cannot be obeyed.} $d$ that is real and non-negative over $x \in \mathbb{R}$, implies that its real roots have even multiplicity, and its complex roots (with non-zero imaginary part) come in complex conjugate pairs. Accordingly, $R(x)$ can be written as 
    \begin{equation}
        R(x) = C \prod_{i=1} (x-r_{i})^{2 \alpha_{r_i}} \prod_{l=1} (x - c_{l})^{\beta_{c_l}} (x - c_{l}^*)^{\beta_{c_l}},
    \end{equation}
    for distinct real roots $r_{i}$ of even multiplicity $2 \alpha_{r_i}$, and distinct complex roots $c_{l}$ of multiplicity $\beta_{c_l}$, and a coefficient $C \in \mathbb{R}$. 
    For real $x \in \mathbb{R}$, $R(x)$ can therefore be expressed as
    \begin{equation}\label{eq:R}
    \begin{aligned}
        &R(x) = \Bigg| \sqrt{C} \prod_{i=1} (x-r_i)^{\alpha_{r_i}} \prod_{l=1} (x-c_l)^{\beta_{c_l}} \Bigg|^2 =: \big| \mathcal{R}(x) \big|^2 ,
    \end{aligned}
    \end{equation}
    where $\mathcal{R}(x)$ is a degree $d/2$ polynomial.

    To show the decomposition stated in this theorem, we need to factorize $\mathcal{R}(x)$ into a product of $k$ factor polynomials: $\mathcal{R}(x) = \prod_{j=1}^k \mathcal{R}_j(x)$. This can be achieved by partitioning the $d/2$ terms in Eq.~\eqref{eq:R} into $k$ groups, and defining $\mathcal{R}_j(x)$ as the product over terms in the $j$th group, times $C^{1/2k}$. Many such groupings exist, so a factorization of $\mathcal{R}(x)$ is not unique. Nonetheless, one can partition the roots such that the first $d/2 \ \text{mod} \ k$ groups are of size $\lfloor d/2k \rfloor + 1$, and the remaining groups are of size $\lfloor d/2k \rfloor$. If $d/2$ divides $k$, then the maximal size is $\lfloor d/2k \rfloor = d/2k$; if $d/2$ does not divide $k$, then the maximal size is $\lfloor d/2k \rfloor + 1 = \lceil d/2k \rceil$. In either case, this guarantees that each factor polynomial has degree at most $\text{deg}(\mathcal{R}_j) \leq \lceil d/2k \rceil$.
\end{proof}

By Lemma~\ref{thm:Real_NonNeg_Poly_Char}, an arbitrary real, non-negative polynomial of even degree $d$ can be decomposed into a product of $k$ factor polynomials squared: $R(x) = \prod_{j=1}^k |\mathcal{R}_j(x)|^2$, where the factor polynomials are of degree at most $\lceil d/2k \rceil = O(d/k)$. This decomposition makes the polynomial $R(x)$ amenable to parallel QSP according to Theorem~\ref{thm:Parallel_QSP}, given that we implement the factor polynomials $\mathcal{R}_j(x)$ with QSP. With this insight, we can characterize the class of polynomials achievable with parallel QSP:

\begin{theorem}[Parallel QSP Polynomial Characterization] \label{thm:parallel_qsp_poly_char}
Let $R(x)$ be a polynomial of even degree $d$, that is real and non-negative over the real axis $x\in \mathbb{R}$. By Lemma~\ref{thm:Real_NonNeg_Poly_Char}, let $R(x)$ factorize into $k$ factor polynomials as $R(x) = \prod_{j=1}^k |\mathcal{R}_j(x)|^2 $, where $\deg(\mathcal{R}_j) \leq \lceil d/2k \rceil$. Invoking parallel QSP across $k$ threads with block-encoded polynomials $\mathcal{R}_j(x)$, we can estimate the trace
\begin{equation}
    z = \tr( \rho^k R(\rho)).
\end{equation}
The requisite query depth is at most $ 2\lceil d/2k \rceil \approx d/k$ and the circuit width is $O(k)$. The number of measurements required to estimate $z$ to additive error $\epsilon$ is 
\begin{equation}
    O \left(\frac{\mathcal{K}\left( R \right)^4 }{\epsilon^2} \right),
\end{equation}
where $\mathcal{K}\left( R \right)$ is a quantity we call the ``factorization constant", whose value depends on the chosen factorization of $R(x)$ as~\footnote{Note that the factorization constant is fundamentally a function of the factor polynomials $\{ \mathcal{R}_j(x) \}$. However, for brevity of notation, we denote it as a function of $R(x)$.}
\begin{equation}\label{eq:factorization_const}
    \mathcal{K}\left( R \right) = \prod_{j=1}^k \|\mathcal{R}_j \|_{[-1,1]}, 
\end{equation}
where the norm $\| \cdot \|_{[-1,1]} $ was defined in Eq.~\eqref{eq:function_norm}. 

\end{theorem}
\begin{proof}
    The density matrix $\rho$ is Hermitian and its eigenvalues real. Therefore, as per Lemma~\ref{thm:Real_NonNeg_Poly_Char}, the action of $R(x)$ on a $\rho$ factorizes as a product of factor polynomials: 
    \begin{equation}
        R(\rho) = \prod_{j=1}^k |\mathcal{R}_j(\rho)|^2,
    \end{equation}
    where $\deg(\mathcal{R}_j) \leq \lceil d/2k \rceil$. Provided block encodings of the factor polynomials, we can apply the results of Theorem~\ref{thm:Parallel_QSP} to extract the trace $\tr(\rho^k \prod_{j=1}^k |\mathcal{R}_j(\rho)|^2 ) = \tr(\rho^k R(\rho))$, as desired. Therefore, all that remains is to construct block encodings of the factor polynomials with QSP.

    However, the factor polynomials do not necessarily obey the conditions of QSP polynomials. For instance, even if $\| R \|_{[-1,1]} \leq 1$, it is not necessarily true that the factor polynomials also obey this condition. In addition, the factor polynomials are in general not of fixed parity. Hence, in full generality, these factor polynomials must be implemented by rescaling by a constant and using a tool like generalized QSP. 

    In more detail, we can implement an arbitrary factor polynomial $\mathcal{R}_j(x)$ by rescaling as
    \begin{equation}\label{eq:factor_poly_rescaling}
        \frac{\mathcal{R}_j(x)}{\| \mathcal{R}_j \|_{[-1,1]}},
    \end{equation}
    which guarantees that this is bounded in magnitude by $1$. We can then block-encode this rescaled polynomial with generalized QSP. As we discussed in Sec.~\ref{sec:Overview_QSP}, the requisite query depth of this procedure is $\leq 2 \lceil d/2k \rceil $. Note however that this simplifies if $\mathcal{R}_j(x)$ is of fixed parity and can be implemented with standard QSP; in this simpler case, the query depth is at most $\lceil d/2k \rceil$.

    Then, to estimate $z = \tr(\rho^k R(\rho) )$, we block encode the $k$ rescaled factor polynomials $\mathcal{R}_j(x)/\| \mathcal{R}_j \|_{[-1,1]}$ in parallel and execute the parallel QSP circuit of Fig.~\ref{fig:Parallel_QSP_Circuit}, which produces an estimate of
    \begin{equation}\label{eq:rescaled_z}
    \begin{aligned}
        \tr(\rho^k \prod_{j=1}^k \Bigg|\frac{\mathcal{R}_j(\rho)}{\| \mathcal{R}_j \|_{[-1,1]}}\Bigg|^2) &=
        \frac{\tr(\rho^k R(\rho))}{\prod_{j=1}^k \| \mathcal{R}_j \|_{[-1,1]}^2 } = \frac{z}{\mathcal{K}(R)^2} ,
    \end{aligned}
    \end{equation}
    where $\mathcal{K}( R ) := \prod_{j=1}^k \| \mathcal{R}_j \|_{[-1,1]}$ is the \emph{factorization constant}, which depends on the chosen factorization of $R(x)$. In order to resolve $z$ to additive error $\epsilon$, it suffices to resolve Eq.~\eqref{eq:rescaled_z} to additive error $\epsilon/\mathcal{K}( R )^2$. According to Theorem~\ref{thm:Parallel_QSP}, this requires a number of measurements $O\left(\mathcal{K}( R )^4 / \epsilon^2 \right)$.
    \end{proof}

Theorem~\ref{thm:parallel_qsp_poly_char} furnishes the following algorithm for estimating the trace $z = \tr(\rho^k R(\rho))$, whose pseudocode we present in Algorithm~\ref{alg:Parallel_QSP}. The first step is to factorize $R(x)$, either analytically or numerically. This can be achieved numerically by determining the roots of $R(x)$ by computing the eigenvalues of its companion matrix~\cite{edelman1995polynomial}; for a degree-$d$ polynomial, this requires $O(d^3)$ time and can be performed as a classical pre-computation step. The next step is to implement the factor polynomials with QSP, and then finally run the parallel QSP circuit as per Theorem~\ref{thm:Parallel_QSP} to obtain an estimate of $z$.

\begin{algorithm}
\caption{Parallel Quantum Signal Processing}\label{alg:Parallel_QSP}
\SetKwInOut{Comp}{Cost}
\SetKwInOut{Proc}{Procedure}
\KwIn{(1) Access to a state $\rho$ and a block encoding of $\rho$; (2) a polynomial $R(x)$ of even degree $d$, that is real and non-negative over $x \in \mathbb{R}$. }
\KwOut{An estimate of $z=\tr(\rho^k R(\rho))$ to additive error $\epsilon$}
\Comp{$O(\mathcal{K}( R )^4 / \epsilon^2)$ executions of a circuit of width $O(k)$ and query depth $O(d/k)$, where $\mathcal{K}( R ) = \prod_{j=1}^k \| \mathcal{R}_j \|_{[-1,1]}$ is a constant that depends on the chosen factorization of $R(x)$.}
\Proc{ }
Classically determine a factorization $R(x) = \prod_{j=1}^k |\mathcal{R}_j(x)|^2$, such that $\deg(\mathcal{R}_j) \leq \lceil d/2k \rceil $ for all $j$\;
Using QSP, construct block encodings of $\mathcal{R}_j(\rho)$ (possibly rescaled as in Eq.~\eqref{eq:factor_poly_rescaling})\;
Run the parallel QSP circuit of Fig.~\ref{fig:Parallel_QSP_Circuit} a number of times $O\left( \mathcal{K}( R )^4/ \epsilon^2 \right)$\,.
\end{algorithm}

We will refer to the number of measurements required by parallel QSP as its \emph{measurement cost}. From Theorem~\ref{thm:parallel_qsp_poly_char}, this is $O\left( \mathcal{K}( R )^4/ \epsilon^2 \right)$, which crucially depends on the chosen factorization of $R(x)$ through the factorization constant $\mathcal{K}(R)$. The factorization constant measures the cost of implementing the factor polynomials, which in general requires rescaling. A poor choice of factorization can result in this constant scaling exponentially in the degree $d$, dashing any quantum advantage provided by this protocol. For instance, by factorizing the order $d$ Chebsyshev polynomial (of the first kind) into two factor polynomials composed of its positive-valued and negative-valued roots respectively, the resulting factorization constant scales as $2^{O(d)}$. 

Therefore, to minimize the measurement cost in practice, it is best to select a \emph{low-norm factorization} of $R(x)$, whose factorization constant scales at worst as $\text{poly}(d) 2^{O(k)} $, rather than $2^{O(d)}$. This can be achieved by selecting factor polynomials of modest norm, which generally requires some analytic knowledge of the polynomial's structure.\footnote{That said, low-norm factorizations are not necessarily rare. For instance, it can be shown that any infinite family of real polynomials whose roots partition the interval $[-1,1]$ into segments of size $O(1/d)$ admit factorizations with norms $\| \mathcal{R}_j \|_{[-1,1]} = \text{poly}(d)$. This can generically be achieved by interleaving the roots of the factor polynomials.}

In addition, as we remarked in the proof of Theorem~\ref{thm:parallel_qsp_poly_char}, while in general the query depth is upper bounded by $2\lceil d/2k \rceil$, if the factor polynomials can all be chosen to be of definite parity, then they can be implemented through standard QSP with a query depth at most $\lceil d/2k \rceil$. Nonetheless, any realization of parallel QSP will attain a query depth scaling as $O(d/k)$. This is because any circuit that prepares a degree-$d$ polynomial $R(\rho)$ requires $O(d)$ instances of $\rho$, whether arranged in parallel or in series. If these instances are parallelized over $k$ threads, the query depth across the threads must be at least $O(d/k)$.

\subsection{Remarks}\label{sec:Remarks}
As presented, the parallel QSP algorithm reduces the query depth needed to compute the trace $z = \tr(\rho^k R(\rho))$. While a standard QSP implementation of the polynomial $\rho^k R(\rho)$ (which has definite parity because $R(x)$ is even) requires a query depth $k+d$, parallel QSP requires a query depth at most $2\lceil d/2k \rceil \approx d/k$. This shrinks the query depth by a factor $\approx k$, yet requires increasing the width to $O(k)$, and the number of measurements by a factor $\mathcal{K}(R)^4$. As $\mathcal{K}(R)$ is a product of $k$ function norms, it generally scales as $2^{O(k)}$. Thus, parallel QSP enables a tradeoff between quantum and classical resources, and is most suitable for platforms limited by short coherence times.

Interestingly, this tradeoff is reminiscent of that encountered in quantum circuit cutting~\cite{Peng_2020, Lowe_2023}. In that context, a quantum circuit is cut across $K$ wires to decompose it into circuits of smaller depth and/or width, which are repeatedly executed to simulate the original quantum circuit. The corresponding measurement overhead scales as $2^{O(K)}$, resembling that of parallel QSP. Likewise, parallel QSP shares similarities with the randomized QSP algorithms presented in Refs.~\cite{tosta2023randomized, wang2024faster, martyn2024halving}. These algorithms randomly sample over QSP polynomials of different degrees and reduce the average degree/query depth. However, while parallel QSP reduces the maximal query depth, these randomized algorithms do not, as high-degree polynomials are still sampled. This makes parallel QSP better suited for quantum hardware constrained by coherence times, where deep circuits are out of reach.

Moreover, the depth-width tradeoff afforded by parallel QSP merits careful explanation. While coherence times constrain circuit depths, they also tend to decrease with the number of qubits, making a tradeoff between depth and width nontrivial in general. Crucially however, in parallel QSP this tradeoff is not arbitrary, but rather is carefully structured: the computation is distributed across $k$ independent threads that do not interact until a final generalized swap test. Because this swap test can be executed in constant depth~\cite{Quek_2024}, the entangled state across all $k$ threads only needs to be maintained briefly. Consequently, the effective coherence requirement for parallel QSP is governed primarily by the coherence time of an individual thread, rather than that of all $k$ threads. Further, the independence of the threads provides a natural form of error resilience: because errors cannot propagate across threads until the swap test, one can simply use error detection to detect and discard erroneous threads, preventing the accumulation of error across the whole computation.

Furthermore, as currently presented, the scope of parallel QSP is limited to computing functions of a density matrix. This follows from the use of the generalized swap test, which enables the multiplication of density matrices arranged in parallel. This limited scope is unsurprising: if an arbitrary QSP algorithm could be parallelized, then one could parallelize Hamiltonian simulation and violate the no-fast forwarding theorem~\cite{Berry_2006} which forbids circuit depths sub-linear in the simulation time~\cite{chia2023impossibility}. Nonetheless, parallel QSP can still be applied to a general operator if it is encoded in a density matrix, such as a Hamiltonian encoded in a thermal state $\rho \propto e^{-\beta H}$ or the state $\rho \propto H+cI$ considered in sample-based Hamiltonian simulation~\cite{Kimmel_2017}.

\section{Parallel QSP for Property Estimation}\label{sec:PQSP_for_Property_Estimation}

A noteworthy application of QSP is estimating properties of a quantum state, expressed as $\tr(f(\rho))$ for a function $f(\rho)$. For instance, the von Neumann entropy is captured by the function $f(\rho) = - \rho \ln \rho$. In practice, such a property can be estimated with QSP by implementing a polynomial $P(\rho) \approx f(\rho)$, and approximating the trace using the techniques of Sec.~\ref{sec:QSP_Trace_Est}. This approach has established algorithms for evaluating the von Neumman entropy~\cite{wang2022new}, R\'enyi entropies~\cite{Subramanian_2021}, fidelities~\cite{gilyen2022improved}, and other related properties.

However, as currently framed, parallel QSP suffers from two limitations that render it inapplicable to general property estimation: parallel QSP (1) applies to a limited class of polynomials $x^k R(x)$ where $R(x)$ is real and non-negative over $x\in \mathbb{R}$, and (2) requires knowledge of a low-norm factorization of $R(x)$ to achieve a reasonable measurement cost (e.g., $\text{poly}(d) 2^{O(k)}$ rather than $2^{O(d)}$). 

In this section we show how to overcome both of these challenges by developing a method that enables parallel QSP to accommodate arbitrary polynomials, while maintaining a query depth $O(d/k)$ and guaranteeing a reasonable measurement cost. This dramatically expands the class of polynomials amenable to parallel QSP, and furnishes property estimation algorithms with reduced query depth.

\subsection{Prelude}
To formalize our problem of interest, suppose we wish to estimate a property by the trace of a real, degree-$d$ polynomial $P(x)$:
\begin{equation}
    w = \tr(P(\rho)), 
\end{equation}
where
\begin{equation}
    P(x) = \sum_{n=0}^d a_n x^n, \ \  \|P\|_{[-1,1]} \leq 1 .
\end{equation}
For a general $P(x)$, estimating $w$ with standard QSP requires a query depth $2d = O(d)$, and also a number of measurements $O(1/\epsilon^2)$ to resolve $w$ with additive error $\epsilon$. In contrast, here we will use parallel QSP to parallelize this computation over $k$ threads, and achieve a query depth $O(d/k + k )$. The resulting measurement cost will depend on the chosen decomposition and factorization of $P(x)$. 

As we remarked above, parallel QSP cannot be directly applied to an arbitrary polynomial $P(x)$. Instead, in order to parallelize the computation over $k$ threads, we split $P(x)$ into a sum of two \textit{constituent polynomials}:
\begin{equation}\label{eq:Poly_decomp}
\begin{aligned}
    P(x) &= \sum_{n=0}^{k-1} a_n x^n + x^k \sum_{n=k}^d a_{n}x^{n-k}  \\
    &=: P_{<k}(x) + x^k P_{\geq k}(x). 
\end{aligned}
\end{equation}
where,
\begin{equation}
\begin{aligned}
    & P_{<k}(x) := \sum_{n=0}^{k-1} a_n x^n, \quad P_{\geq k}(x) := \sum_{n=0}^{d-k} a_{n+k}x^{n}, 
\end{aligned}
\end{equation}
are the constituent polynomials of $P(x)$. $P_{<k}(x) $ and $P_{\geq k }(x)$ are real polynomials of degree $k-1$ and $d-k$, respectively. With this decomposition, the desired property can be written as a sum of two \emph{constituent traces}
\begin{equation}\label{eq:Constituent_traces}
\begin{aligned}
    & w = w_{<k} + w_{\geq k}, \\
    & w_{<k} = \tr(P_{< k} (\rho)), \  w_{\geq k} = \tr(\rho^k P_{\geq k} (\rho)). 
\end{aligned}
\end{equation}
Therefore, to estimate $w$, it will equivalently suffice to estimate $w_{<k}$ and $w_{\geq k}$.

Importantly, $w_{<k}$ is the trace of a polynomial of degree $k-1$, which can be easily estimated with standard QSP at query depth $2(k-1)$ and width $O(1)$. On the other hand, $w_{\geq k}$ is the trace of $\rho^k$ times a polynomial of degree $d-k$, which nearly fits into the framework of parallel QSP. By incorporating appropriate algebraic manipulations to ensure that $P_{\geq k} (x)$ is non-negative, we will estimate $w_{\geq k}$ with parallel QSP at a query depth $\approx (d-k)/k < d/k$ and circuit width $O(k)$. Therefore, the overall requisite query depth to estimate $w$ is guaranteed to never exceed $\approx \max \{ 2k, d/k \}$. In practice, $k \ll d$ (e.g., a large degree polynomial parallelized over a few threads), in which case the query depth reduces to $ \approx d/k$.

In this manner, property estimation with parallel QSP can be viewed as a hybrid of standard and parallel QSP, where the low degree terms are estimated with standard QSP, and the higher degree terms with parallel QSP. Below, we will investigate this procedure more closely. We first consider the case in which $P_{\geq k} (x)$ is non-negative and hence directly amenable to parallel QSP. We next show that even if $P_{\geq k} (x)$ is not non-negative, it can be decomposed into a basis of non-negative polynomials and thus made amenable to parallel QSP. In both situations, we include bounds on the requisite query depth and measurement costs to estimate $w$ to a desired level of error. In the first situation, the cost crucially depends on the factorization of $P_{\geq k} (x)$ as per Theorem~\ref{thm:parallel_qsp_poly_char}. In the second case, the measurement cost depends on the factorization of our purported decomposition; we prove the existence of a low-norm factorization such that this contribution scales as $O(d^4 2^{O(k)}/k^2)$.

Lastly, as we will see in the following, although $P(x)$ is bounded as $\| P \|_{[-1,1]} \leq 1$, the constituent polynomials are not necessarily bounded the same: $\| P_{< k} \|_{[-1,1]}, \| P_{\geq k}  \|_{[-1,1]} \nleq 1$. As a result, the measurement cost of estimating $w$ will also depend on the norms $\| P_{< k} \|_{[-1,1]}$ and $\| P_{\geq k}  \|_{[-1,1]}$. We prove in Appendix~\ref{app:bounds_partial_sum_polys} that for any bounded polynomial $P(x)$, its constituent polynomials are upper bounded as 
\begin{equation}\label{eq:magnitude_bounds}
    \begin{aligned}
        & \sup_{P(x), \ \|P \|_{[-1,1]} \leq 1} \big\|P_{< k} \big\|_{[-1,1]} \leq O\Bigg( \frac{d^{k-1}}{(k-1)!} \Bigg), \\
        & \sup_{P(x), \ \|P \|_{[-1,1]} \leq 1} \big\| P_{\geq k}  \big\|_{[-1,1]} \leq O\Bigg( \frac{d^k}{k!} \sqrt{\frac{d}{k}}\Bigg) . 
    \end{aligned}
\end{equation}
Therefore, for $k=O(1)$, it is necessarily the case that $\|P_{< k} \|_{[-1,1]}, \|P_{< k} \|_{[-1,1]} = O({\rm poly}(d))$ scales at worst as a polynomial in $d$. Moreover, we emphasize that these are worst case bounds, and that many polynomials of interest have constituent polynomials with much smaller norms. For example, the polynomial approximation to the exponential function $e^{-\beta (x+1)}$ (e.g., for thermal state preparation) has constituent polynomials both upper bounded in magnitude by $O(1)$, independent of $d$.

\subsection{Estimation by Direct Application of Parallel QSP}
If $P_{\geq k}(x)$ is non-negative, then we can use the above intuition to estimate the property $w  =\tr(P(\rho))$ by a direct application of parallel QSP, and achieve a query depth $\approx d/k $:
\begin{theorem}[Parallel QSP for Property Estimation: Direct Application] \label{thm:parallel_qsp_prop_est_1}
Consider a real polynomial $P(x)$ of degree $d$, that is bounded as $\| P \|_{[-1,1]} \leq 1$. Let $P(x)$ decompose according to Eq.~\eqref{eq:Poly_decomp} as 
\begin{equation}
    P(x) = P_{< k} (x) + x^k P_{\geq k} (x),
\end{equation}
and suppose that $P_{\geq k} (x)$ is non-negative over $x \in \mathbb{R}$. By invoking parallel QSP across $k$ threads, we can estimate $w = \tr(P(\rho))$ with query depth at most ${\rm max} \{ 2(k-1)  , 2 \lceil \tfrac{d-k}{2k} \rceil \} \lesssim {\rm max} \{ 2k  , d/k \}  = O(d/k + k)$. The number of measurements required to resolve $w$ to additive error $\epsilon$ is 
\begin{equation}
    O \Bigg( \frac{ \big\| P_{< k}  \big\|_{[-1,1]}^2 +  \mathcal{K}\left(P_{\geq k}  \right)^4  }{\epsilon^2 } \Bigg) ,  
\end{equation}
where $ \mathcal{K}\left(P_{\geq k}  \right) $ is the factorization constant defined in Eq.~\eqref{eq:factorization_const} and depends on the chosen factorization of $P_{\geq k}  (x)$. 
\end{theorem}
\begin{proof}
    Decompose $w=w_{<k}+w_{\geq k}$ as per Eq.~\eqref{eq:Constituent_traces}. We then want to estimate $w_{<k}$ and $w_{\geq k}$ each to error $\leq \epsilon/2$, such that their sum estimates $w$ with error at most $\epsilon$. 

    First, one can estimate $w_{<k}$ to error $\epsilon/2$ with standard QSP, using for instance a Hadamard test. Because $P_{< k} (x)$ is a real polynomial of degree $k-1$ and indefinite parity, this can be achieved with query depth $2(k-1)$, and a number of measurements $O \left( \| P_{< k}  \|_{[-1,1]}^2 / \epsilon^2  \right) $. 

    Next, if $P_{\geq k} (x)$ is non-negative over $x\in \mathbb{R}$, then the constituent trace $w_{\geq k}$ obeys the conditions of Theorem~\ref{thm:parallel_qsp_poly_char} and can be directly estimated with parallel QSP. Because $P_{\geq k} (x)$ is a polynomial of degree $d-k$, this can be achieved at a query depth $2 \lceil \tfrac{d-k}{2k} \rceil$, and a number of measurements $O \left( \mathcal{K}( P_{\geq k}  )^4 / \epsilon^2 \right) $. 
    
    In aggregate, the maximal query depth is $\max \{ 2(k-1), 2 \lceil \tfrac{d-k}{2k} \rceil \}$, and the total number of measurements is 
    \begin{equation}
    \begin{aligned}
        &O \left( \frac{ \big\| P_{< k}  \big\|_{[-1,1]}^2}{ \epsilon^2 } \right) + O \left( \frac{\mathcal{K}( P_{\geq k}  )^4}{\epsilon^2 }\right) . 
    \end{aligned}
    \end{equation} 
\end{proof}

Evidently, the measurement cost of Theorem~\ref{thm:parallel_qsp_prop_est_1} depends on the quantities $\| P_{< k}  \|_{[-1,1]}$ and $\mathcal{K}(P_{\geq k} )$, which necessarily depend on the polynomial $P(x)$ under consideration. According to the bounds of Eq.~\eqref{eq:magnitude_bounds}, we can guarantee that for any such polynomial, $\| P_{< k}  \|_{[-1,1]} \leq \text{poly}(d)$, implying that this term grows polynomially in $d$, even in the worst case. On the other hand, as we mentioned in Sec.~\ref{sssec:parallel_QSP_alg}, minimizing $\mathcal{K}(P_{\geq k} )$ requires determining a low-norm factorization of $P_{\geq k} (x)$, whose factorization constant is not prohibitively large.

For numerical intuition, we have developed code to implement parallel QSP according to Theorem~\ref{thm:parallel_qsp_prop_est_1}. Provided a number of threads $k$ and a polynomial $P(x)$, this code decomposes $P(x)$ into constituent polynomials $P_{< k} (x)$ and $P_{\geq k} (x)$, factors $P_{\geq k}(x)$ across $k$ threads, and then determines the corresponding QSP phases of the factor polynomials. Our code can be found on GitHub at Ref.~\cite{parallel_QSP}; see Appendix~\ref{sec:dicussion_code_implementation} for more details.


\subsection{Estimation by Parallel QSP and Decomposition into a Well-Behaved Basis}

For an arbitrary polynomial $P(x)$, the constituent polynomial $P_{\geq k} (x)$ is not necessarily non-negative, which renders Theorem~\ref{thm:parallel_qsp_prop_est_1} inapplicable. Likewise, while simply factorizing $P_{\geq k}(x)$ can be done efficiently, guaranteeing that this corresponds to a modest factorization constant is more difficult, and in general requires knowledge about the structure of $P_{\geq k} (x)$. If either of these criteria are unsatisfied, then as we show here, parallel QSP can still be applied by taking special pre-computation steps to reduce the problem back to Theorem~\ref{thm:parallel_qsp_prop_est_1}.

We achieve this by decomposing $P_{\geq k} (x)$ into a basis of polynomials that are each amenable to parallel QSP and known to admit a low-norm factorization. We additionally prove that the measurement cost incurred by this factorization scales as $O(d^4 2^{O(k)}/k^2)$, which crucially maintains polynomial scaling in $d$. This enables parallelization of a large class of property estimation problems, and furnishes the main result of this paper:
\begin{theorem}[Parallel QSP for Arbitrary Property Estimation: Definite Parity]\label{thm:parallel_qsp_prop_est_3}
Let $P(x)$ be a real polynomial of degree $d$ and definite parity, that is bounded as $\| P \|_{[-1,1]} \leq 1$. By invoking parallel QSP across $k$ threads, where $k$ has the same parity as $d$, we can estimate
\begin{equation}
    w = \tr(P(\rho))
\end{equation}
with a circuit of width $O(k)$ and query depth at most $\lfloor \frac{d-k}{2k} \rfloor + k-1 = O(d/k+k)$. The number of measurements required to resolve $w$ to additive error $\epsilon$ is 
\begin{equation}
\begin{aligned}
    &O\Bigg( \frac{\|P_{< k} \|_{[-1,1]}^2}{\epsilon^2} + \frac{\|P_{\geq k} \|_{[-1,1]}^2 d^4 \big( 1+\sqrt{2} \big)^{4k} }{k^2 \epsilon^2 } \Bigg) \\
    &= O\Bigg( \frac{\|P_{< k} \|_{[-1,1]}^2 + \|P_{\geq k} \|_{[-1,1]}^2 d^4 2^{O(k)}}{\epsilon^2} \Bigg).
\end{aligned}
\end{equation}
\end{theorem}

For brevity, we defer the full proof of this theorem to Appendix~\ref{app:ParallelQSPPropertyEstimation}. As an overview, the proof works by first decomposing $P_{\geq k} (x)$ into the basis of Chebyshev polynomials. By then using properties of the Chebyshev polynomials (specifically, their composition and product relations), we can re-express $P_{\geq k} (x)$ as a linear combination of products of squared Chebyshev polynomials. These products are each amenable to parallel QSP and clearly exhibit a low-norm factorization, contributing a factor of $O(d^4 2^{O(k)}/k^2 )$ to the measurement cost. In addition, as the Chebyshev polynomials are real and of definite parity, they can be implemented directly with QSP and correspond to a query depth $\approx d/2k$, in contrast to the the query depth $\approx d/k$ for general polynomials according Theorem~\ref{thm:parallel_qsp_prop_est_1}.

We can also extend Theorem~\ref{thm:parallel_qsp_prop_est_3} to polynomials of indefinite parity, which yields an analogous result: 
\begin{theorem}[Parallel QSP for Arbitrary Property Estimation: Indefinite Parity] \label{thm:parallel_qsp_prop_est_2}
    Let $P(x)$ be a real polynomial of degree $d$, that is bounded as $\| P \|_{[-1,1]} \leq 1$. By invoking parallel QSP across $k$ threads, we can estimate
    \begin{equation}
        w = \tr(P(\rho))
    \end{equation}
    with a circuit of width $O(k)$ and query depth at most $\big\lfloor \frac{d-k}{2(k-1)} \big\rfloor + k-2 \approx d/2k + k = O(d/k+k)$. The number of measurements required to resolve $w$ to additive error $\epsilon$ is 
    \begin{equation}
    \begin{aligned}
        &O\Bigg( \frac{\|P_{< k} \|_{[-1,1]}^2}{\epsilon^2} + \frac{\|P_{\geq k} \|_{[-1,1]}^2 d^4 \big( 1+\sqrt{2} \big)^{4k} }{k^2 \epsilon^2 } \Bigg) \\
        & = O\Bigg( \frac{\|P_{< k} \|_{[-1,1]}^2 + \|P_{\geq k} \|_{[-1,1]}^2 d^4 2^{O(k)}}{\epsilon^2} \Bigg).
    \end{aligned}
    \end{equation} 
\end{theorem}
    
\begin{proof}
    Decompose $P(x)$ into its even and odd components: $ P(x) = P_{\text{even}}(x) + P_{\text{odd}}(x)$. These are both bounded as $\| P_{\text{even}} \|_{[-1,1]}, \  \| P_{\text{odd}} \|_{[-1,1]} \leq \| P \|_{[-1,1]} \leq 1$ by the triangle inequality. Therefore, we can apply Theorem~\ref{thm:parallel_qsp_prop_est_3} to estimate the traces $\tr(P_{\text{even}}(\rho))$ and $\tr(P_{\text{odd}}(\rho))$ each to additive error $\epsilon/2$, such that their sum approximates $\tr(P(\rho))$ to error $\epsilon$. 
    
    However, Theorem~\ref{thm:parallel_qsp_prop_est_3} requires that $k$ have the same parity as the polynomial whose trace is being estimated. This is not possible for both $P_{\text{even}}(x)$ and $P_{\text{odd}}(x)$ given only a single value of $k$. Thus, if $k$ is even (odd), then estimate $\tr(P_{\text{even}}(\rho))$  over $k$ ($k-1$) threads and $\tr(P_{\text{odd}}(\rho))$ over $k-1$ ($k$) threads. This amounts to replacing $k$ with $k-1$ in the requisite query depth and number of measurements. As per Theorem~\ref{thm:parallel_qsp_prop_est_3}, this corresponds to a query depth at most $\big\lfloor \frac{d-k+1}{2(k-1)} \big\rfloor + k-2 = O(d/k+k)$, and a total number of measurements 
    \begin{equation}
    \begin{aligned}
        &O\Bigg( \frac{\|P_{< k} \|_{[-1,1]}^2 + \|P_{\geq k} \|_{[-1,1]}^2 d^4 2^{O(k)}}{\epsilon^2} \Bigg).
    \end{aligned}
    \end{equation} 
\end{proof}
    
Ultimately, Theorem~\ref{thm:parallel_qsp_prop_est_2} is applicable to any real, bounded polynomial $P(x)$, with a dependence only on the norms of the constituent polynomials. As a result, this theorem captures most properties of interest and renders parallel QSP applicable to a broad class of problems. In the following section, we will see that in problems of interest, this norm scales as a polynomial in $d$ or even remains constant. For completeness, we can assess the worst-case overhead (corresponding to the most ill-behaved polynomial) by inserting the general bounds Eq.~\eqref{eq:magnitude_bounds}, which equate to a worst-case overhead $\| P_{\geq k} \|_{[-1,1]}^2 = O\big(\frac{d^{2k+1}}{k \cdot (k!)^2} \big) $. While this expression grows as a degree-$2k+1$ polynomial in $d$, it also decays quadratic-factorially in $k$, ultimately dominating the otherwise $2^{O(k)}$ factor in the measurement overhead. We now turn to applications of parallel QSP in property estimation.

\section{Applications} \label{sec:applications}
Here we use parallel QSP to develop parallelized algorithms for various problems in property estimation. We highlight the estimation of R\'enyi entropy, general polynomials, partition functions, and the von Neumman entropy.

\subsection{R\'enyi Entropy: Integer Order}

One of the most straightforward demonstrations of parallel QSP is the computation of the R\'enyi entropy. For a state $\rho$, the R\'enyi entropy of order $\alpha$ is defined as $S_\alpha(\rho) = \frac{1}{1-\alpha}\log\big(\tr(\rho^\alpha) \big)$ for $\alpha>0$, $\alpha \neq 1$. The R\'enyi entropy provides a probe of entanglement in both quantum and classical simulation~\cite{Brydges_2019, Elben_2018, Linke_2018, Hastings_2010, Hibat_Allah_2020}, and can be used to approximate the spectrum of $\rho$ through entanglement spectroscopy~\cite{Johri_2017}. 

Let us first consider estimating $S_\alpha(\rho)$ for integer orders $\alpha \geq 2$, deferring non-integer orders to Sec.~\ref{sec:Renyi_Nonint}. In this case, prior work has introduced QSP-based algorithms that implement $\rho^\alpha$ as a QSP polynomial, enabling the estimation of $S_\alpha(\rho)$ with query depth $\alpha$ and width $O(1)$~\cite{Subramanian_2021, wang2022new, Wang_2023}. On the other hand, Refs.~\cite{Hastings_2010, Johri_2017, Quek_2024} invoke the generalized swap test across $\alpha$ systems to evaluate $S_\alpha(\rho)$, corresponding to a query depth $O(1)$ and width $O(\alpha)$. Here, we will illustrate how parallel QSP interpolates between these two regimes, achieving query depth $O(\alpha/k)$ and width $O(k)$.

To demonstrate this, note that the polynomial of interest here is $P(\rho) = \rho^\alpha $. This is a monomial that trivially factorizes into a product of smaller monomials. Following this logic, we arrive at the following theorem:
\begin{theorem}[Parallel QSP for Estimating R\'enyi Entropy: Integer Order]\label{thm:PQSP_Renyi_Int}
    For positive integers $\alpha \geq 2$, one can invoke parallel QSP across $k$ to estimate the R\'enyi entropy $S_\alpha(\rho) = \frac{1}{1-\alpha}\log\big(\tr(\rho^\alpha) \big)$ with a circuit of query depth $\lfloor \frac{1}{k}\lfloor \frac{\alpha-k}{2}\rfloor \rfloor+1 = O(\alpha/k)$ and width $O(k)$. The number of measurements required to achieve additive error $\epsilon$ is $O\Big(\frac{1}{s_\alpha^2 \alpha^2 \epsilon^2} \Big)$ where $s_\alpha = \tr(\rho^\alpha)$. 
\end{theorem}
\begin{proof}
    Decompose $P(\rho)$ according to Eq.~\eqref{eq:Poly_decomp} as
    \begin{equation}\label{eq:Renyi_decomp}
        P(\rho) = \rho^\alpha = \rho^k \rho^{(\alpha - k) \bmod 2} \big| \rho^{\lfloor \frac{\alpha-k}{2} \rfloor} \big|^2, 
    \end{equation} 
    where we assume $\alpha > k$ (otherwise there is no need to parallelize). This nearly fits into the family of polynomials amenable to parallel QSP as per Theorem~\ref{thm:Parallel_QSP}, the only difference being the additional factor $\rho^{(\alpha - k) \bmod 2}$. This factor is only relevant if $(\alpha - k) \bmod 2 = 1$, in which case it can be easily incorporated by including an additional thread in the initial state $\rho$ into the generalize swap test stage, which increases the threads to $k+1$ and still equates to a width $O(k)$.

    This decomposition corresponds to constituent polynomials $P_{< k} (\rho) = 0$ and $P_{\geq k} (\rho) = \rho^{\lfloor \frac{\alpha-k}{2}\rfloor}$. While $P_{< k}(\rho)$ is non-existent, $P_{\geq k} (\rho)$ is bounded as $\| P_{\geq k} \|_{[-1,1]} \leq 1$ and factorizes into a product of monomials:
    \begin{equation}\label{eq:Renyi_factorization}
    \begin{aligned}
        & \rho^{\lfloor \frac{\alpha-k}{2} \rfloor} = \\ 
        & \quad \prod_{j=1}^{\lfloor \frac{\alpha-k}{2}\rfloor \bmod k} \rho^{\lfloor \frac{1}{k}\lfloor \frac{\alpha-k}{2}\rfloor \rfloor +1 } \times \prod_{j'=\lfloor \frac{\alpha-k}{2}\rfloor \bmod k + 1}^{k} \rho^{\lfloor \frac{1}{k}\lfloor \frac{\alpha-k}{2}\rfloor \rfloor} . 
    \end{aligned}
    \end{equation} 
    These factor polynomials are all real and of definite parity, and thus can be implemented directly via QSP with a query depth at most $\lfloor \frac{1}{k}\lfloor \frac{\alpha-k}{2}\rfloor \rfloor+1 = O(\alpha/k)$. Because these monomials are bounded in magnitude by $1$, the corresponding factorization constant is simply $\mathcal{K}(P_{\geq k}) = 1$. Plugging these values into Theorem~\ref{thm:parallel_qsp_prop_est_1}, the measurement cost of estimating $s_\alpha = \tr(\rho^\alpha)$ to additive error $\epsilon'$ is $O(1/\epsilon'^2)$. 

    However, our quantity of interest is $S_\alpha (\rho ) = \frac{1}{1-\alpha} \log( s_\alpha )$. Ref.~\cite{Subramanian_2021} shows that an estimate $\tilde{s}_\alpha$ of $s_\alpha$ to within multiplicative error $\varepsilon$ (i.e., $|\tilde{s}_\alpha / s_\alpha - 1| \leq \varepsilon$) provides an approximation to $S_\alpha (\rho )$ with additive error $\varepsilon/(\alpha-1)$:
    \begin{equation}\label{eq:Renyi_err_bound}
        \Big|\frac{1}{1-\alpha} \log( \tilde{s}_\alpha ) - S_\alpha(\rho) \Big| \leq \frac{\varepsilon}{\alpha - 1}.
    \end{equation}
    Therefore, to achieve additive error $\epsilon$ in the estimate of the R\'enyi entropy, select $\epsilon' =  s_\alpha \varepsilon = s_\alpha \epsilon (\alpha-1)$, which translates to a total number of measurements 
    \begin{equation}
        O\Big(\frac{1}{s_\alpha^2 \alpha^2 \epsilon^2 } \Big). 
    \end{equation}
\end{proof}

As a hybrid of QSP and the generalized swap test, parallel QSP combines the strengths of both algorithms in estimating the R\'enyi entropy. Its query depth $O(\alpha/k)$ and width $O(k)$ provide a smooth interpolation between the circuit requirements of these approaches, allowing for full use of one's available quantum resources. 

\subsection{General Polynomials in the Monomial Basis}

Above, we used parallel QSP to parallelize the computation of the trace of a monomial $\tr(\rho^\alpha)$. This naturally extends to more general polynomials $P(\rho) = \sum_{n=0}^d c_j \rho^n$ by parallelizing each monomial. Following this line of thought, we can prove the following:
\begin{theorem}[Parallel QSP for Estimation of General Polynomial Traces]\label{thm:PQSP_GeneralPoly}
    For a degree-$d$ polynomial $P(x) = \sum_{n=0}^d c_n x^n $, one can invoke parallel QSP across $k$ threads to estimate $\tr(P(\rho))$ with a circuit of query depth $\lfloor \frac{1}{k}\lfloor \frac{d-k}{2}\rfloor \rfloor+1 = O(d/k)$ and width $O(k)$. The number of measurements required to attain additive error $\epsilon$ is 
    \begin{equation}
        O\left(\frac{\| c \|_1^2}{\epsilon^2} \right),
    \end{equation}
    where $\| c \|_1 = \sum_{n=0}^d |c_n|$ is the 1-norm of the polynomial coefficients. 
\end{theorem}
\begin{proof}
    Our aim is to parallelize the computation of $\tr(P(\rho)) = \sum_{n=0}^d c_n \tr(\rho^n) $ by using the method of Theorem~\ref{thm:PQSP_Renyi_Int} for parallelizing the monomial trace $\tr(\rho^n)$. While one could achieve this by estimating each trace $\tr(\rho^n)$ in sequence (from $n=0$ up to $n=d$), a more efficient approach is furnished by importance sampling. That is, sample an integer $n$ from the distribution $p(n) = |c_n|/\| c \|_1 $ where $\| c\|_1 = \sum_{n=0}^d |c_n| $, and evaluate the estimator of $\tr(\rho^n)$ (i.e. the measurement result of the parallel QSP circuit). As we show in Appendix~\ref{app:Importance_Sampling}, this procedure provides an estimator for ${\tr(P(\rho))}/{\| c \|_1} $. Then, estimating $\tr(P(\rho))$ to additive error $\epsilon$ is equivalent to estimating ${\tr(P(\rho))}/{\| c \|_1} $ to error $\epsilon/\| c \|_1$, which requires a measurement cost $O(\| c \|_1^2/\epsilon^2)$. 

    In measuring the estimator of each monomial trace $\tr(\rho^n)$, employ parallel QSP according to Theorem~\ref{thm:PQSP_Renyi_Int}. The corresponding query depth is at most $\lfloor \frac{1}{k}\lfloor \frac{d-k}{2}\rfloor \rfloor+1 = O(d/k)$, and the circuit width is $O(k)$. 
\end{proof}

As the measurement cost grows with $\|c\|_1$, this approach is best suited for polynomials where this 1-norm is not prohibitively large. For many polynomials of interest, $\|c\|_1$ is $\text{poly}(d)$ or even just $O(1)$. For instance, a truncation of $e^{-\beta x}$ achieves $\|c\|_1 = O(e^\beta)$, independent of $d$. However, $\|c\|_1$ can grow exponentially with $d$ for certain polynomials, in which case this approach is prohibitively costly. For example, the order $d$ Chebyshev polynomial $T_d(x)$ has 1-norm $\| c\|_1 = 2^{O(d)}$. In this case, it is better to invoke Theorem~\ref{thm:parallel_qsp_prop_est_2} to estimate $\tr(P(\rho))$, which applies to all polynomials agnostic to their $1$-norm, and achieves a cost that necessarily scales polynomially in $d$.

While estimation of functions in the monomial basis serves as a pedagogical introduction to the mechanics of parallel QSP (after all, factoring monomials is trivial!), the true strength of our approach lies in handling polynomials with large 1-norm $\| c\|_1$ in the monomial basis, such as the Chebyshev polynomials $T_d(x)$. In such cases, prior techniques~\cite{Quek_2024, Yirka_2021}, including Theorem~\ref{thm:PQSP_GeneralPoly}, incur a cost that scales with $\| c\|_1$ and can be exponentially large in $d$. In contrast, parallel QSP, as explained in Theorem~\ref{thm:parallel_qsp_prop_est_2}, applies a polynomial transformation coherently, with a complexity determined by a function norm that is often significantly smaller than $\|c\|_1$. This makes parallel QSP especially powerful for property estimation tasks where the target polynomial is not sparse or has large coefficients, such as those used in estimating the non-integer R\'enyi entropy or the von Neumann entropy. We explore these polynomials later in Secs.~\ref{sec:Renyi_Nonint} and~\ref{sec:VonNeumann}.

\subsection{Partition Function}
A problem to which Theorem~\ref{thm:PQSP_GeneralPoly} applies nicely is the estimation of $\tr(e^{-\beta \rho})$. This can be used to evaluate a partition function $Z = \tr(e^{-\beta H})$ when the Hamiltonian $H$ is encoded in a density matrix as $\rho \propto H+cI$, such as in sample-based Hamiltonian simulation~\cite{Kimmel_2017} or density matrix exponentiation~\cite{Lloyd_2014}. Here we can prove the following: 
\begin{theorem}[Estimation of $\tr(e^{-\beta \rho})$ with Parallel QSP]
    One an invoke parallel QSP across $k$ threads to estimate $\tr(e^{-\beta \rho})$ with a circuit of query depth $O\left(\frac{\beta}{k} \log(\beta D/\epsilon) \right)$ and width $O(k)$, where $D = \text{dim}(\rho)$. To achieve additive error $\epsilon$, the requisite number of measurements is $O\big( \frac{e^{2\beta}}{\epsilon^2} \big)$.  
\end{theorem}
\begin{proof}
    Consider the polynomial $P(x) = \sum_{n=0}^d \frac{(-\beta)^n}{n!}x^n$. This approximates the exponential $e^{-\beta x}$ with an additive error at most $\epsilon'$ over $x\in[0,1]$ by choosing a degree $d = O(\beta \log(\beta/\epsilon'))$~\cite{Berry_2015_Simulating}. As our goal is to evaluate $\tr(P(\rho))$ to error $\epsilon$, select a polynomial error $\epsilon' = \epsilon/2D$, where $D=\text{dim}(\rho)$. This guarantees that $|\tr(P(\rho)) - \tr(e^{-\beta \rho})| \leq |\tr(\epsilon ')| = \epsilon/2$. Therefore, if we can evaluate $\tr(P(\rho))$ to error $\epsilon/2$, we will approximate $\tr(e^{-\beta \rho})$ to error $\epsilon$. 

    To perform this evaluation, we first note the 1-norm bound:
    \begin{equation}
        \| c\|_1 = \sum_{n=0}^d \frac{\beta^n}{n!} < e^\beta. 
    \end{equation}
    According to Theorem~\ref{thm:PQSP_GeneralPoly}, we can estimate $\tr(P(\rho))$ with a circuit of query depth $O(d/k) = O(\frac{\beta}{k}\log(\beta D/\epsilon))$ and width $O(k)$, and the requisite number of measurements is 
    \begin{equation}
        O \Bigg( \frac{\| c \|_1^2 }{(\epsilon/2)^2} \Bigg) = O \Big( \frac{e^{2\beta}}{\epsilon^2} \Big). 
    \end{equation} 
\end{proof}

Relative to standard QSP, parallel QSP reduces the query depth by a factor $O(k)$, without worsening the scaling of the measurement cost. This is because $e^{-\beta x}$ admits a Taylor series whose coefficients have a 1-norm bounded by $O(1)$, independent of the truncation degree.

\subsection{R\'enyi Entropy: Non-Integer Order}\label{sec:Renyi_Nonint}
It is also of interest to compute the R\'enyi entropy for non-integer $\alpha$. Rather straightforwardly, this may be acheived by implementing $\rho^\alpha$ as a QSP polynomial and estimating its trace~\cite{Subramanian_2021, wang2022new}. This however is more costly to compute than the case of integer order because the function $x^\alpha$ is non-analytic at $x=0$, and hence can only be well approximated by a polynomial away from this singular point. Due to this singularity, polynomial approximations to $x^\alpha$ have coefficients whose 1-norm scales exponentially as $\|c\|_1 = 2^{O(d)}$~\cite{Quek_2024}, thus rendering Theorem~\ref{thm:PQSP_GeneralPoly} impractical. In this case, it is advantageous to instead use Theorem~\ref{thm:parallel_qsp_prop_est_2}, whose measurement cost is guaranteed to scale polynomially in $d$.

We can use this reasoning to prove the following theorem for estimating the R\'enyi entropy. For generality, we state this theorem for a degree-$d$ polynomial approximation of $x^\alpha$. As we will discuss after the proof, the necessary degree depends on properties of $\rho$ (condition number, rank, etc.).
\begin{theorem}[Parallel QSP for Estimating R\'enyi Entropy: Non-Integer Order]\label{thm:PQSP_Renyi_NonInt}
    Suppose $\alpha>0$ is a non-integer, and let $P(x)$ be a degree-$d$ polynomial approximation to $x^\alpha$. Then, one can invoke parallel QSP across $k$ threads to estimate the R\'enyi entropy $S_\alpha(\rho) = \frac{1}{1-\alpha}\log\big(\tr(\rho^\alpha) \big)$ as $\tilde{S}_\alpha(\rho) = \frac{1}{1-\alpha}\log\big(\tr(P(\rho)) \big)$, with a circuit of query depth at most $\big\lfloor \frac{d-k}{2(k-1)} \big\rfloor + k-2 = O(d/k+k)$ and width $O(k)$. The number of measurements necessary to estimate $\tilde{S}_\alpha(\rho)$ to additive error $\epsilon$ is 
    \begin{equation}\label{eq:RenyiEstMeasCost}
        O\Bigg( \Bigg(\frac{d^{k+2}}{(k+1)!} \Bigg)^2 \frac{d}{k} \frac{1}{\tilde{s}_\alpha^2 \epsilon^2 (\alpha-1)^2} \Bigg) , 
    \end{equation}
    where $\tilde{s}_\alpha = \tr(P(\rho)) \approx \tr(\rho^\alpha)$. 
\end{theorem}
\begin{proof}
    Without a specific polynomial approximation to $x^\alpha$ (which we will address after this proof)  nor an analytic factorization thereof, we can invoke the generic construction of Theorem~\ref{thm:parallel_qsp_prop_est_2}, which is applicable to all polynomials. This reduces the query depth to $\big\lfloor \frac{d-k}{2(k-1)} \big\rfloor + k-2 = O(k+d/k)$. Inserting the generic bounds of Eq.~\eqref{eq:magnitude_bounds}, we find that the measurement cost to estimate $\tilde{s}_\alpha = \tr(P(\rho))$ to additive error $\epsilon'$ is
    \begin{equation}
    \begin{aligned}
        &O\Bigg( \Bigg(\frac{d^{k}}{k!} \sqrt{\frac{d}{k}} \Bigg)^2  \frac{d^4}{k^2} \frac{(1+\sqrt{2})^{4k}}{\epsilon'^2} \Bigg) \\
        & \qquad \qquad \qquad \qquad \qquad = O\Bigg( \Bigg(\frac{d^{k+2}}{(k+1)!}\Bigg)^2 \frac{d}{k \epsilon'^2} \Bigg) . 
    \end{aligned}
    \end{equation}
    where the second line follows from the factorial $k!$ dominating the exponential $2^{O(k)}$. Then, mirroring the proof of Theorem~\ref{thm:PQSP_Renyi_Int}, in order to estimate $\tilde{S}_\alpha(\rho)$ to additive error $\epsilon$, it suffices to select $\epsilon' = s_\alpha \varepsilon = s_\alpha \epsilon (\alpha-1)$, which equates to a total number of measurements 
    \begin{equation}
        O\Bigg( \Bigg(\frac{d^{k+2}}{(k+1)!} \Bigg)^2 \frac{d}{k} \frac{1}{\tilde{s}_\alpha^2 \epsilon^2 (\alpha-1)^2} \Bigg). 
    \end{equation}

    Lastly, note that we have used the generic bounds of Eq.~\eqref{eq:magnitude_bounds}. Although these are worst-case, and could possibly be tightened for a specific polynomial approximation, we expect them to be relatively tight for polynomial approximations to functions with singularities like $x^\alpha$. 
\end{proof}

As anticipated, parallel QSP reduces the depth by a factor of $O(k)$, and increases the number of measurements by a polynomial in $d$ and a factor that decays quadratic-factorially in $k$. To precisely approximate the R\'enyi entropy using this result, we need to determine a sufficient polynomial approximation to $x^\alpha$. Refs.~\cite{Wang_2023, Gily_n_2019} provide an odd polynomial approximation to $x^\alpha$ that suffers additive error at most $\epsilon$ over $x \in [\delta,1]$ for some $\delta>0$, and show that the degree of this polynomial is $O(\alpha + \frac{1}{\delta}\log(1/\epsilon))$. Using this polynomial, there are two predominant approaches for R\'enyi entropy estimation. 

First, as in Ref.~\cite{Subramanian_2021}, one can choose $\delta$ to be the smallest non-zero eigenvalue of $\rho$ (or a lower bound thereof), such that $P(\rho) \approx \rho^\alpha$ over the support of $\rho$. This choice is equivalent to setting $\delta = 1/\kappa$ where $\kappa$ is the condition number of $\rho$. Then, to estimate $S_\alpha(\rho)$ to error $\epsilon$, we can select $\epsilon' = s_\alpha \epsilon |\alpha-1|/2D$ where $D = \text{dim}(\rho)$, corresponding to a degree $d = O(\alpha + \kappa \log(D/|\alpha-1| s_\alpha \epsilon))$. This choice guarantees that $|\tilde{s}_\alpha - s_\alpha| = |\tr(P(\rho)) - \tr(\rho^\alpha) | \leq \tr(\epsilon') = s_\alpha \epsilon |\alpha-1|/2$, such that $|\tilde{S}_\alpha(\rho) - S_\alpha(\rho)| \leq \epsilon/2$ by Eq.~\eqref{eq:Renyi_err_bound}. Therefore, to estimate $S_\alpha(\rho)$ to error $\epsilon$, it suffices to estimate $\tilde{S}_\alpha(\rho)$ to error $\epsilon/2$. According to Theorem~\ref{thm:PQSP_Renyi_NonInt}, the corresponding parallel QSP algorithm requires a number of measurements 
\begin{equation}\label{eq:condition_num_meas}
\begin{aligned}
    &O\Bigg( \Bigg(\frac{d^{k+2}}{(k+1)!} \Bigg)^2 \frac{d}{k} \frac{1}{\tilde{s}_\alpha^2 \epsilon^2 (\alpha-1)^2} \Bigg) \\
    & \qquad \qquad \qquad \qquad \qquad = O\Bigg( \frac{\text{poly}(\alpha, \kappa, \log(D))}{s_\alpha^2 \epsilon^2 (\alpha-1)^2} \Bigg) . 
\end{aligned}
\end{equation}

On the other hand, if the rank $r=\text{rank}(\rho)$ is known, then as Ref.~\cite{wang2022new} shows, one can alternatively select a value $\delta = O(1/r)$ while maintaining an accurate approximation of the R\'enyi entropy. Here, one can again select $\epsilon' = s_\alpha \epsilon |\alpha-1|/2r$, such that $|\tilde{s}_\alpha - s_\alpha| = |\tr(P(\rho)) - \tr(\rho^\alpha) | \leq \tr(\epsilon') = s_\alpha \epsilon |\alpha-1|/2$, implying $|\tilde{S}_\alpha(\rho) - S_\alpha(\rho)| \leq \epsilon/2$ by Eq.~\eqref{eq:Renyi_err_bound}. These choices correspond to a polynomial degree $d = O(\alpha + r \log(r/s_\alpha \epsilon |\alpha-1|))$. As above, it suffices to estimate $\tilde{S}_\alpha(\rho)$ to error $\epsilon/2$, in which case Theorem~\ref{thm:PQSP_Renyi_NonInt} indicates that parallel QSP necessitates a number of measurements 
\begin{equation}
\begin{aligned}
    & O\Bigg( \frac{\text{poly}(\alpha, r)}{s_\alpha^2 \epsilon^2 (\alpha-1)^2} \Bigg) . 
\end{aligned}
\end{equation}
This scales more favorably than the above approach leading to Eq.~\eqref{eq:condition_num_meas}, because $r \leq \kappa$ for any density matrix.

In both cases, the query depth is reduced relative to the associated sequential algorithms~\cite{Subramanian_2021, wang2022new, Wang_Quantum_2023}, at the expense of increasing the measurement cost by a factor $O(\text{poly}(\alpha, \kappa, \log(D)))$ or $O(\text{poly}(\alpha, r))$. As such, these methods are best deployed on states with small condition number or low-rank, e.g. $\kappa, r = O(\text{polylog}(D))$, which arise in a variety of experimental~\cite{Gross_2010, Butucea_2015, Araujo_2024}, computational~\cite{Bridgeman_2017, perezgarcia2007matrix, Jarzyna_2013}, and theoretical contexts~\cite{Lloyd_2014, ezzell2022quantum}. In these cases, quantum algorithms for R\'enyi entropy estimation achieve an exponential advantage over classical algorithms, which scale as $O(\text{poly}(D))$~\cite{wang2022new, Wang_Quantum_2023}. Fortunately, because the measurement cost of parallel QSP scales polynomially in $\kappa$ and $r$, parallel QSP crucially retains this exponential speedup.

\subsection{Von Neumann Entropy}\label{sec:VonNeumann}
Another ubiquitous quantity in quantum physics is the von Neumann entropy $S(\rho) = -\tr(\rho \ln \rho)$, which characterizes entanglement~\cite{nielsen2010quantum}, defines thermal states~\cite{chowdhury2020variational}, describes phase transitions~\cite{Skinner_2019}, and dictates the rate at which information can be transmitted across a quantum channel~\cite{Schumacher_1995}. Given such broad interest, QSP-based algorithms have been put forth to estimate the von Neumann entropy~\cite{wang2022new, Wang_2023}, which work by approximating $-\rho \ln \rho$ as a QSP polynomial and taking its trace. Similar to the situation of the non-integer $\alpha$-R\'enyi entropy, the underlying function $-x\ln x$ is singular at $x=0$, and can only be approximated by a polynomial away from this point. As such, polynomial approximations to $-x\ln x$ also exhibit coefficients whose 1-norm scales exponentially as $\|c\|_1 = 2^{O(d)}$~\cite{Quek_2024}, thus rendering Theorem~\ref{thm:PQSP_GeneralPoly} impractical and suggests the use of  Theorem~\ref{thm:parallel_qsp_prop_est_2} instead.

We can therefore invoke parallel QSP to prove the following theorem, analogous to Theorem~\ref{thm:PQSP_Renyi_NonInt}. As above, we state this for a degree-$d$ polynomial approximation of $-x\ln x$, where $d$ will depend on the properties $\rho$, which we address after the proof.
\begin{theorem}[Parallel QSP for Estimating von Neumann Entropy]\label{thm:PQSP_VonNeumann}
    Let $P(x)$ be a degree-$d$ polynomial approximation to  $-x\ln x$. Then, parallel QSP enables the estimation of the von Neumann entropy $S(\rho) = -\tr(\rho \ln \rho)$ as $\tilde{S}(\rho) = \tr(P(\rho))$, with a circuit of query depth at most $\big\lfloor \frac{d-k}{2(k-1)} \big\rfloor + k-2 = O(d/k + k)$ and width $O(k)$. The number of measurements necessary to estimate $\tilde{S}(\rho)$ to additive error $\epsilon$ is 
    \begin{equation}\label{eq:VonNeumannEstMeasCost}
        O\Bigg( \Bigg(\frac{d^{k+2}}{(k+1)!} \Bigg)^2 \frac{d}{k} \frac{1}{\epsilon^2} \Bigg). 
    \end{equation}
\end{theorem}
\begin{proof}
    This proof is nearly identical to Theorem~\ref{thm:PQSP_Renyi_NonInt}. In the absence of a specific polynomial approximation to $-x\ln(x)$ nor a factorization thereof, we again apply Theorem~\ref{thm:parallel_qsp_prop_est_2} to estimate $\tr(P(\rho))$ to additive error $\epsilon$, and insert the generic bounds of Eq.~\eqref{eq:magnitude_bounds}. This reduces the query depth to $\big\lfloor \frac{d-k}{2(k-1)} \big\rfloor + k - 2 = O(d/k+k)$ and necessitates a measurement cost
    \begin{equation}
    \begin{aligned}
        &O\Bigg( \Bigg(\frac{d^{k}}{k!} \sqrt{\frac{d}{k}} \Bigg)^2  \frac{d^4}{k^2} \frac{(1+\sqrt{2})^{4k}}{\epsilon^2} \Bigg) = \\
        & O\Bigg( \Bigg(\frac{d^{k+2}}{(k+1)!}\Bigg)^2 \frac{d}{k \epsilon^2} \Bigg) . 
    \end{aligned}
    \end{equation}
    Again, we expect the generic bounds to be nearly saturated for polynomial approximations to a singular function like $-x \ln x$. 
\end{proof}

As above, parallel QSP reduces the query depth by a factor $O(k)$ and increases the number of measurements by a factor $\text{poly}(d)$. What remains is to provide a sufficient polynomial approximation to $-x\ln x$. Similar to the function $x^\alpha$, there exists an odd polynomial approximation to to $-x\ln x$ that suffers additive error at most $\epsilon '$ over $x\in[\delta,1]$ for some $\delta>0$, and the requisite degree of this polynomial is $O(\frac{1}{\delta} \log(1/\epsilon'))$~\cite{wang2022new, Wang_2023, Gily_n_2019}.

Analogous to our treatment of the R\'enyi entropy, one can use this polynomial and select $\delta = 1/\kappa$ and $\epsilon' = \epsilon/2D$, where $D = \text{dim}(\rho)$. This guarantees that $|\tilde{S}(\rho) - S(\rho) | \leq \epsilon/2$ and corresponds to a degree $d = O(\kappa \log(D/\epsilon))$. It then suffices to estimate $\tilde{S}(\rho)$ to error $\epsilon/2$, in which case Theorem~\ref{thm:PQSP_VonNeumann} indicates that parallel QSP requires query depth $O(d/k+k)$ and a number of measurements $O(\frac{\text{poly}(\kappa, \log D)}{\epsilon^2})$. 

Alternatively, if the rank $r=\text{rank}(\rho)$ is known, one can select $\delta = O(1/r)$ and $\epsilon' = O(\epsilon/r)$, while maintaining an accurate approximation to the von Neumann entropy~\cite{wang2022new}. This equates to a degree $d=O(r\log(r/\epsilon))$. Then, Theorem~\ref{thm:PQSP_VonNeumann} readily implies that parallel QSP estimates $S(\rho)$ with query depth $O(d/k+k)$ and a number of measurements $O(\frac{\text{poly}(r)}{\epsilon^2})$.

Again, these methods are best deployed on states with $\kappa, r = O(\text{polylog}(D))$, in which case quantum algorithms provide an exponential improvement over classical algorithms~\cite{wang2022new, Wang_Quantum_2023}. Because the measurement cost of parallel QSP scales as $\text{poly}(\kappa, \log(D) )$ or $\text{poly}(r)$, parallel QSP retains this exponential speedup.

\section{Discussion and Conclusion} \label{sec:discussion}

In this work, we have presented a parallelized version of quantum signal processing, tailored to property estimation. Our algorithm parallelizes the computation of a property $\tr(P(\rho))$ over $k$ threads, and reduces query depth by a factor $O(k)$ while increasing circuit width by $O(k)$, characterizing a tradeoff between temporal and spatial resources in QSP. The core of our construction rests on the ability to factorize a polynomial of degree $d$ into a product of $k$ degree-$O(d/k)$ polynomials. Within a quantum circuit, these polynomials are prepared in parallel and then multiplied together using a generalized swap test. This methodology is applicable to general polynomials of a density matrix, with a measurement cost that depends on the chosen factorization. In justification of its utility, we apply parallel QSP to a variety of problems in property estimation, ranging from the estimation of entropies to partition functions.

Our work has important implications for intermediate-scale quantum computation. First, because parallel QSP reduces circuit depth, it is well suited for platforms limited to shallow circuits by short coherence times~\cite{Somoroff_2023}. Secondly, parallel QSP opens up new opportunities for co-designing parallel quantum architectures, with applications in distributed quantum computing~\cite{caleffi2022distributed} and resource-intensive tasks on near-term devices~\cite{dgn_24_low_depth_qsp}. In this context, separate devices perform QSP in parallel and then send their state to a central device that performs a swap test, with the key benefit that errors remain isolated across devices until the swap test. Lastly, parallel QSP can help to reduce the overhead incurred by error mitigation or error correction. For example, in the structured, algorithmic-level correction of coherent errors in QSP~\cite{tan2023error}, parallel QSP reduces the overhead in circuit depth from $O(d^n)$ to $O((\frac{d}{k})^n)$ when errors are corrected to order $n$.


There exist numerous possible improvements to our parallel QSP algorithm. One limitation is the measurement overhead accrued when parallelizing over a generic polynomial, which does not always permit a factorization into polynomials of bounded magnitude on an interval. While here we circumvented this issue by decomposing into a well-behaved basis, it remains an open problem to bound the factorization constant for a generic polynomial and provide tighter bounds on the resulting measurement overhead. 
In addition, a theory of QSP on SU$(d)$ may prove useful in combining multiple monomials together \cite{laneve2023quantum,lu2024quantum} with correspondingly better overhead. 
With an eye toward improving property estimation algorithms, it may also prove fruitful to integrate parallel QSP with randomized measurements techniques~\cite{Elben_2018, Elben_2022}.

Parallel QSP suggests multiple future directions for the design of further parallel quantum algorithms. First, while we discussed parallel QSP in the context of univariate polynomials, one could extend this algorithm to multivariate polynomials~\cite{Quek_2024} to encompass relevant metrics like the fidelity $\tr(\sqrt{ \sqrt{\rho} \sigma \sqrt{\rho}}) $ and trace distances $\frac{1}{2} \tr(|\rho-\sigma|^n)$. Although parallel QSP can straightforwardly accommodate products of univariate polynomials by direct substitution into the parallel QSP circuit (Fig.~\ref{fig:Parallel_QSP_Circuit}), designing a generic multivariate polynomial is a more pressing challenge as a general theory of multivariate QSP has yet to be established \cite{rossi_m_qsp_22,nemeth2023variants, gomes2024multivariable}.
Secondly, beyond polynomial factorization, other forms of parallelization are also possible. For example, by exploring spatial locality, one could use spline interpolations to approximate a well-behaved polynomial by lower-degree polynomials over restricted intervals; implementation of this idea may require access to projectors that divide the input space into restricted intervals. 
In addition, we note that there exists concurrent work considering depth reduction of QSP sequences, albeit in the setting of parameter estimation for calibration \cite{dgn_24_low_depth_qsp}. Understanding the relation of our methods to such ideas remains an interesting open question for future work.

Together, these prospective generalizations of parallel QSP, combined with parallel classical computing \cite{alexeev2024quantum}, offer utility in other areas, such as the parallel quantum simulation of quantum chemical systems and materials \cite{liu2023bootstrap, bauer2020quantum}. Fundamentally, as properties of polynomial factorizations underlies parallel QSP, this suggests that other functional analytic properties could be leveraged to develop new and improved quantum algorithms through QSP. Given that polynomials have been extensively studied for centuries, boasting a rich array of properties from complex analysis to algebraic geometry, further exploration of the connections between polynomials and QSP may point to novel quantum algorithms.

\textit{Acknowledgements}: The authors thank Patrick Rall, Danial Motlagh, and Alexander Zlokapa for useful discussion. JMM acknowledges support from the National Science Foundation Graduate Research Fellowship under Grant No. 2141064. This project was supported by the U.S. Department of Energy, Office of Science, National Quantum Information Science Research Centers, Co-design Center for Quantum Advantage (C$^2$QA) under contract number DE-SC0012704. YL also acknowledges support from the Department of Energy under contract number DE-SC0025384. This research was also supported in part by NTT Research.

\bibliography{References}

\begin{thebibliography}{10}

\bibitem{ylc_fixed_point_14}
Theodore~J. Yoder, Guang~Hao Low, and Isaac~L. Chuang.
\newblock ``Fixed-point quantum search with an optimal number of queries''.
\newblock \href{https://dx.doi.org/https://doi.org/10.1103/PhysRevLett.113.210501}{Phys. Rev. Lett. {\bf 113}, 210501}~(2014).

\bibitem{Low_2016}
Guang~Hao Low, Theodore~J. Yoder, and Isaac~L. Chuang.
\newblock ``Methodology of resonant equiangular composite quantum gates''.
\newblock \href{https://dx.doi.org/http://dx.doi.org/10.1103/PhysRevX.6.041067}{Physical Review X~{\bf 6}}~(2016).

\bibitem{Low_2017}
Guang~Hao Low and Isaac~L. Chuang.
\newblock ``Optimal hamiltonian simulation by quantum signal processing''.
\newblock \href{https://dx.doi.org/http://dx.doi.org/10.1103/PhysRevLett.118.010501}{Physical Review Letters~~{\bf 118}}~(2017).

\bibitem{Low_2019}
Guang~Hao Low and Isaac~L. Chuang.
\newblock ``Hamiltonian simulation by qubitization''.
\newblock \href{https://dx.doi.org/10.22331/q-2019-07-12-163}{Quantum {\bf 3}, 163}~(2019).

\bibitem{Gily_n_2019}
Andr{\'a}s Gily{\'e}n, Yuan Su, Guang~Hao Low, and Nathan Wiebe.
\newblock ``Quantum singular value transformation and beyond: exponential improvements for quantum matrix arithmetics''.
\newblock \href{https://dx.doi.org/http://dx.doi.org/10.1145/3313276.3316366}{Proceedings of the 51st Annual ACM SIGACT Symposium on Theory of Computing}~(2019).

\bibitem{martyn2021grand}
John~M Martyn, Zane~M Rossi, Andrew~K Tan, and Isaac~L Chuang.
\newblock ``Grand unification of quantum algorithms''.
\newblock \href{https://dx.doi.org/https://doi.org/10.1103/PRXQuantum.2.040203}{PRX Quantum {\bf 2}, 040203}~(2021).

\bibitem{Haah_2019}
Jeongwan Haah.
\newblock ``Product decomposition of periodic functions in quantum signal processing''.
\newblock \href{https://dx.doi.org/10.22331/q-2019-10-07-190}{Quantum {\bf 3}, 190}~(2019).

\bibitem{chao2020finding}
Rui Chao, Dawei Ding, Andras Gily{\'e}n, Cupjin Huang, and Mario Szegedy.
\newblock ``Finding angles for quantum signal processing with machine precision''~(2020).
\newblock  \href{http://arxiv.org/abs/2003.02831}{arXiv:2003.02831}.

\bibitem{Dong_2021}
Yulong Dong, Xiang Meng, K.~Birgitta Whaley, and Lin Lin.
\newblock ``Efficient phase-factor evaluation in quantum signal processing''.
\newblock \href{https://dx.doi.org/10.1103/physreva.103.042419}{Physical Review A~{\bf 103}}~(2021).

\bibitem{dlnw_infinite_qsp_22}
Yulong Dong, Lin Lin, Hongkang Ni, and Jiasu Wang.
\newblock ``Infinite quantum signal processing''.
\newblock \href{https://dx.doi.org/10.22331/q-2024-12-10-1558}{Quantum {\bf 8}, 1558}~(2024).

\bibitem{ms_query_comp_fun_24}
Ashley Montanaro and Changpeng Shao.
\newblock ``Quantum and classical query complexities of functions of matrices''.
\newblock In Proceedings of the 56th Annual ACM Symposium on Theory of Computing.
\newblock \href{https://dx.doi.org/10.1145/3618260.3649665}{Page 573–584}.
\newblock STOC 2024~New York, NY, USA~(2024). Association for Computing Machinery.

\bibitem{cs_qsvt_tang_tian_23}
Ewin Tang and Kevin Tian.
\newblock ``A {CS} guide to the quantum singular value transformation''.
\newblock \href{https://dx.doi.org/10.1137/1.9781611977936.13}{Chapter : 2024 Symposium on Simplicity in Algorithms (SOSA), pages 121--143}.
\newblock SIAM. ~(2024).

\bibitem{almtw_szego_qsp_24}
Michel Alexis, Lin Lin, Gevorg Mnatsakanyan, Christoph Thiele, and Jiasu Wang.
\newblock ``Infinite quantum signal processing for arbitrary {S}zeg\"{o} functions''~(2024).
\newblock  \href{http://arxiv.org/abs/2407.05634}{arXiv:2407.05634}.

\bibitem{berntson2024complementary}
Bjorn~K. Berntson and Christoph Sünderhauf.
\newblock ``Complementary polynomials in quantum signal processing''~(2024).
\newblock  \href{http://arxiv.org/abs/2406.04246}{arXiv:2406.04246}.

\bibitem{debry_qsp_23}
Kyle DeBry, Jasmine Sinanan-Singh, Colin~D. Bruzewicz, David Reens, May~E. Kim, Matthew~P. Roychowdhury, Robert McConnell, Isaac~L. Chuang, and John Chiaverini.
\newblock ``Experimental quantum channel discrimination using metastable states of a trapped ion''.
\newblock \href{https://dx.doi.org/10.1103/PhysRevLett.131.170602}{Phys. Rev. Lett. {\bf 131}, 170602}~(2023).

\bibitem{kikuchi2023realization}
Yuta Kikuchi, Conor Mc~Keever, Luuk Coopmans, Michael Lubasch, and Marcello Benedetti.
\newblock ``Realization of quantum signal processing on a noisy quantum computer''.
\newblock \href{https://dx.doi.org/https://doi.org/10.1038/s41534-023-00762-0}{npj Quantum Information~ {\bf 9}, 93}~(2023).

\bibitem{ps_architecture_11}
David~A Patterson, Frederick~P Brooks~Jr, Ivan~E Sutherland, and Charles~P Thacker.
\newblock ``Computer architecture''.
\newblock Elsevier Science. ~(2011).

\bibitem{Johri_2017}
Sonika Johri, Damian~S. Steiger, and Matthias Troyer.
\newblock ``Entanglement spectroscopy on a quantum computer''.
\newblock \href{https://dx.doi.org/10.1103/physrevb.96.195136}{Physical Review B~{\bf 96}}~(2017).

\bibitem{Elben_2018}
A.~Elben, B.~Vermersch, M.~Dalmonte, J.I. Cirac, and P.~Zoller.
\newblock ``Rényi entropies from random quenches in atomic hubbard and spin models''.
\newblock \href{https://dx.doi.org/10.1103/physrevlett.120.050406}{Physical Review Letters~{\bf 120}}~(2018).

\bibitem{Buhrman_2001}
Harry Buhrman, Richard Cleve, John Watrous, and Ronald de~Wolf.
\newblock ``Quantum fingerprinting''.
\newblock \href{https://dx.doi.org/10.1103/physrevlett.87.167902}{Physical Review Letters~{\bf 87}}~(2001).

\bibitem{oszmaniec_relational_24}
Micha{\l} Oszmaniec, Daniel~J Brod, and Ernesto~F Galv{\~a}o.
\newblock ``Measuring relational information between quantum states, and applications''.
\newblock \href{https://dx.doi.org/https://doi.org/10.1088/1367-2630/ad1a27}{New Journal of Physics {\bf 26}, 013053}~(2024).

\bibitem{low2017quantum}
Guang~Hao Low.
\newblock ``Quantum signal processing by single-qubit dynamics''.
\newblock PhD thesis.
\newblock Massachusetts Institute of Technology.
\newblock ~(2017).
\newblock  url:~\url{http://hdl.handle.net/1721.1/115025}.

\bibitem{Quek_2024}
Yihui Quek, Eneet Kaur, and Mark~M. Wilde.
\newblock ``Multivariate trace estimation in constant quantum depth''.
\newblock \href{https://dx.doi.org/10.22331/q-2024-01-10-1220}{Quantum {\bf 8}, 1220}~(2024).

\bibitem{tan2023error}
Andrew~K. Tan, Yuan Liu, Minh~C. Tran, and Isaac~L. Chuang.
\newblock ``Perturbative model of noisy quantum signal processing''.
\newblock \href{https://dx.doi.org/10.1103/PhysRevA.107.042429}{Phys. Rev. A {\bf 107}, 042429}~(2023).

\bibitem{rycs_noisy_qsp_22}
Zane~M Rossi, Jeffery Yu, Isaac~L Chuang, and Sho Sugiura.
\newblock ``Quantum advantage for noisy channel discrimination''.
\newblock \href{https://dx.doi.org/https://doi.org/10.1103/PhysRevA.105.032401}{Physical Review A~ {\bf 105}, 032401}~(2022).

\bibitem{caleffi2022distributed}
Marcello Caleffi, Michele Amoretti, Davide Ferrari, Jessica Illiano, Antonio Manzalini, and Angela~Sara Cacciapuoti.
\newblock ``Distributed quantum computing: A survey''.
\newblock \href{https://dx.doi.org/https://doi.org/10.1016/j.comnet.2024.110672}{Computer Networks {\bf 254}, 110672}~(2024).

\bibitem{quek2022exponentially}
Yihui Quek, Daniel Stilck~França, Sumeet Khatri, Johannes~Jakob Meyer, and Jens Eisert.
\newblock ``Exponentially tighter bounds on limitations of quantum error mitigation''.
\newblock \href{https://dx.doi.org/10.1038/s41567-024-02536-7}{Nature Physics {\bf 20}, 1648--1658}~(2024).

\bibitem{Peng_2020}
Tianyi Peng, Aram~W. Harrow, Maris Ozols, and Xiaodi Wu.
\newblock ``Simulating large quantum circuits on a small quantum computer''.
\newblock \href{https://dx.doi.org/10.1103/physrevlett.125.150504}{Physical Review Letters~{\bf 125}}~(2020).

\bibitem{motlagh2023generalized}
Danial Motlagh and Nathan Wiebe.
\newblock ``Generalized quantum signal processing''.
\newblock \href{https://dx.doi.org/10.1103/PRXQuantum.5.020368}{PRX Quantum {\bf 5}, 020368}~(2024).

\bibitem{dong2021efficient}
Yulong Dong, Xiang Meng, K~Birgitta Whaley, and Lin Lin.
\newblock ``Efficient phase-factor evaluation in quantum signal processing''.
\newblock \href{https://dx.doi.org/https://doi.org/10.1103/PhysRevA.103.042419}{Physical Review A~ {\bf 103}, 042419}~(2021).

\bibitem{Ying_2022}
Lexing Ying.
\newblock ``Stable factorization for phase factors of quantum signal processing''.
\newblock \href{https://dx.doi.org/10.22331/q-2022-10-20-842}{Quantum {\bf 6}, 842}~(2022).

\bibitem{yamamoto2024robust}
Shuntaro Yamamoto and Nobuyuki Yoshioka.
\newblock ``Robust angle finding for generalized quantum signal processing''~(2024).
\newblock  \href{http://arxiv.org/abs/2402.03016}{arXiv:2402.03016}.

\bibitem{Subramanian_2021}
Sathyawageeswar Subramanian and Min-Hsiu Hsieh.
\newblock ``Quantum algorithm for estimating $\alpha$ -renyi entropies of quantum states''.
\newblock \href{https://dx.doi.org/10.1103/physreva.104.022428}{Physical Review A~{\bf 104}}~(2021).

\bibitem{gilyen2019distributional}
András Gilyén and Tongyang Li.
\newblock ``Distributional property testing in a quantum world''~(2019).
\newblock  \href{http://arxiv.org/abs/1902.00814}{arXiv:1902.00814}.

\bibitem{aharonov2006polynomial}
Dorit Aharonov, Vaughan Jones, and Zeph Landau.
\newblock ``A polynomial quantum algorithm for approximating the jones polynomial''.
\newblock In Proceedings of the Thirty-Eighth Annual ACM Symposium on Theory of Computing.
\newblock \href{https://dx.doi.org/10.1145/1132516.1132579}{Page 427–436}.
\newblock STOC '06~New York, NY, USA~(2006). Association for Computing Machinery.

\bibitem{gilyen2022improved}
András Gilyén and Alexander Poremba.
\newblock ``Improved quantum algorithms for fidelity estimation''~(2022).
\newblock  \href{http://arxiv.org/abs/2203.15993}{arXiv:2203.15993}.

\bibitem{wang2022new}
Qisheng Wang, Ji~Guan, Junyi Liu, Zhicheng Zhang, and Mingsheng Ying.
\newblock ``New quantum algorithms for computing quantum entropies and distances''.
\newblock \href{https://dx.doi.org/10.1109/TIT.2024.3399014}{IEEE Transactions on Information Theory {\bf 70}, 5653--5680}~(2024).

\bibitem{nielsen2010quantum}
Michael~A Nielsen and Isaac~L Chuang.
\newblock ``Quantum computation and quantum information''.
\newblock Cambridge University Press. ~(2010).

\bibitem{Brassard_2002}
Gilles Brassard, Peter Høyer, Michele Mosca, and Alain Tapp.
\newblock ``Quantum amplitude amplification and estimation''.
\newblock \href{https://dx.doi.org/10.1090/conm/305/05215}{Quantum Computation and Information~Page 53–74}~(2002).

\bibitem{Ekert_2002}
Artur~K. Ekert, Carolina~Moura Alves, Daniel K.~L. Oi, Michał Horodecki, Paweł Horodecki, and L.~C. Kwek.
\newblock ``Direct estimations of linear and nonlinear functionals of a quantum state''.
\newblock \href{https://dx.doi.org/10.1103/physrevlett.88.217901}{Physical Review Letters~{\bf 88}}~(2002).

\bibitem{Hastings_2010}
Matthew~B. Hastings, Iván González, Ann~B. Kallin, and Roger~G. Melko.
\newblock ``Measuring renyi entanglement entropy in quantum monte carlo simulations''.
\newblock \href{https://dx.doi.org/10.1103/physrevlett.104.157201}{Physical Review Letters~{\bf 104}}~(2010).

\bibitem{brun2004measuring}
Todd~A. Brun.
\newblock ``Measuring polynomial functions of states''~(2004).
\newblock  \href{http://arxiv.org/abs/quant-ph/0401067}{arXiv:quant-ph/0401067}.

\bibitem{Horodecki_2002}
Paweł Horodecki and Artur Ekert.
\newblock ``Method for direct detection of quantum entanglement''.
\newblock \href{https://dx.doi.org/10.1103/physrevlett.89.127902}{Physical Review Letters~{\bf 89}}~(2002).

\bibitem{Suba__2019}
Yiğit Subaşı, Lukasz Cincio, and Patrick~J Coles.
\newblock ``Entanglement spectroscopy with a depth-two quantum circuit''.
\newblock \href{https://dx.doi.org/10.1088/1751-8121/aaf54d}{Journal of Physics A: Mathematical and Theoretical {\bf 52}, 044001}~(2019).

\bibitem{Yirka_2021}
Justin Yirka and Yiğit Subaşı.
\newblock ``Qubit-efficient entanglement spectroscopy using qubit resets''.
\newblock \href{https://dx.doi.org/10.22331/q-2021-09-02-535}{Quantum {\bf 5}, 535}~(2021).

\bibitem{edelman1995polynomial}
Alan Edelman and Hiroshi Murakami.
\newblock ``Polynomial roots from companion matrix eigenvalues''.
\newblock \href{https://dx.doi.org/https://doi.org/10.1090/S0025-5718-1995-1262279-2}{Mathematics of Computation {\bf 64}, 763--776}~(1995).

\bibitem{Lowe_2023}
Angus {Lowe}, Matija {Medvidovi{\'c}}, Anthony {Hayes}, Lee~J. {O'Riordan}, Thomas~R. {Bromley}, Juan~Miguel {Arrazola}, and Nathan {Killoran}.
\newblock ``Fast quantum circuit cutting with randomized measurements''.
\newblock \href{https://dx.doi.org/10.22331/q-2023-03-02-934}{Quantum {\bf 7}, 934}~(2023).

\bibitem{tosta2023randomized}
Allan Tosta, Thais de~Lima~Silva, Giancarlo Camilo, and Leandro Aolita.
\newblock ``Randomized semi-quantum matrix processing''.
\newblock \href{https://dx.doi.org/10.1038/s41534-024-00883-0}{npj Quantum Information~{\bf 10}}~(2024).

\bibitem{wang2024faster}
Yue Wang and Qi~Zhao.
\newblock ``Faster quantum algorithms with "fractional"-truncated series''~(2024).
\newblock  \href{http://arxiv.org/abs/2402.05595}{arXiv:2402.05595}.

\bibitem{martyn2024halving}
John~M. Martyn and Patrick Rall.
\newblock ``Halving the cost of quantum algorithms with randomization''.
\newblock \href{https://dx.doi.org/10.1038/s41534-025-01003-2}{npj Quantum Information~{\bf 11}}~(2025).

\bibitem{Berry_2006}
Dominic~W. Berry, Graeme Ahokas, Richard Cleve, and Barry~C. Sanders.
\newblock ``Efficient quantum algorithms for simulating sparse hamiltonians''.
\newblock \href{https://dx.doi.org/10.1007/s00220-006-0150-x}{Communications in Mathematical Physics {\bf 270}, 359–371}~(2006).

\bibitem{chia2023impossibility}
Nai-Hui Chia, Kai-Min Chung, Yao-Ching Hsieh, Han-Hsuan Lin, Yao-Ting Lin, and Yu-Ching Shen.
\newblock ``On the impossibility of general parallel fast-forwarding of hamiltonian simulation''.
\newblock In Proceedings of the Conference on Proceedings of the 38th Computational Complexity Conference.
\newblock \href{https://dx.doi.org/10.4230/LIPIcs.CCC.2023.33}{CCC '23~}Dagstuhl, DEU~(2023). Schloss Dagstuhl--Leibniz-Zentrum fuer Informatik.

\bibitem{Kimmel_2017}
Shelby Kimmel, Cedric Yen-Yu Lin, Guang~Hao Low, Maris Ozols, and Theodore~J. Yoder.
\newblock ``Hamiltonian simulation with optimal sample complexity''.
\newblock \href{https://dx.doi.org/10.1038/s41534-017-0013-7}{npj Quantum Information~{\bf 3}}~(2017).

\bibitem{parallel_QSP}
Kevin Cheng.
\newblock ``{Parallel} {QSP}''.
\newblock \url{https://github.com/kevinchengg/parallelQSP}~(2024).

\bibitem{Brydges_2019}
Tiff Brydges, Andreas Elben, Petar Jurcevic, Benoît Vermersch, Christine Maier, Ben~P. Lanyon, Peter Zoller, Rainer Blatt, and Christian~F. Roos.
\newblock ``Probing rényi entanglement entropy via randomized measurements''.
\newblock \href{https://dx.doi.org/10.1126/science.aau4963}{Science {\bf 364}, 260–263}~(2019).

\bibitem{Linke_2018}
N.~M. Linke, S.~Johri, C.~Figgatt, K.~A. Landsman, A.~Y. Matsuura, and C.~Monroe.
\newblock ``Measuring the rényi entropy of a two-site fermi-hubbard model on a trapped ion quantum computer''.
\newblock \href{https://dx.doi.org/10.1103/physreva.98.052334}{Physical Review A~{\bf 98}}~(2018).

\bibitem{Hibat_Allah_2020}
Mohamed Hibat-Allah, Martin Ganahl, Lauren~E. Hayward, Roger~G. Melko, and Juan Carrasquilla.
\newblock ``Recurrent neural network wave functions''.
\newblock \href{https://dx.doi.org/10.1103/physrevresearch.2.023358}{Physical Review Research~{\bf 2}}~(2020).

\bibitem{Wang_2023}
Youle Wang, Lei Zhang, Zhan Yu, and Xin Wang.
\newblock ``Quantum phase processing and its applications in estimating phase and entropies''.
\newblock \href{https://dx.doi.org/10.1103/physreva.108.062413}{Physical Review A~{\bf 108}}~(2023).

\bibitem{Lloyd_2014}
Seth Lloyd, Masoud Mohseni, and Patrick Rebentrost.
\newblock ``Quantum principal component analysis''.
\newblock \href{https://dx.doi.org/10.1038/nphys3029}{Nature Physics {\bf 10}, 631–633}~(2014).

\bibitem{Berry_2015_Simulating}
Dominic~W. Berry, Andrew~M. Childs, Richard Cleve, Robin Kothari, and Rolando~D. Somma.
\newblock ``Simulating hamiltonian dynamics with a truncated {T}aylor series''.
\newblock \href{https://dx.doi.org/http://dx.doi.org/10.1103/PhysRevLett.114.090502}{Physical Review Letters~{\bf 114}}~(2015).

\bibitem{Wang_Quantum_2023}
Youle Wang, Benchi Zhao, and Xin Wang.
\newblock ``Quantum algorithms for estimating quantum entropies''.
\newblock \href{https://dx.doi.org/10.1103/PhysRevApplied.19.044041}{Phys. Rev. Appl. {\bf 19}, 044041}~(2023).

\bibitem{Gross_2010}
David Gross, Yi-Kai Liu, Steven~T. Flammia, Stephen Becker, and Jens Eisert.
\newblock ``Quantum state tomography via compressed sensing''.
\newblock \href{https://dx.doi.org/10.1103/physrevlett.105.150401}{Physical Review Letters~{\bf 105}}~(2010).

\bibitem{Butucea_2015}
Cristina Butucea, Mădălin Guţă, and Theodore Kypraios.
\newblock ``Spectral thresholding quantum tomography for low rank states''.
\newblock \href{https://dx.doi.org/10.1088/1367-2630/17/11/113050}{New Journal of Physics {\bf 17}, 113050}~(2015).

\bibitem{Araujo_2024}
Israel~F. Araujo, Carsten Blank, Ismael C.~S. Araújo, and Adenilton~J. da~Silva.
\newblock ``Low-rank quantum state preparation''.
\newblock \href{https://dx.doi.org/10.1109/tcad.2023.3297972}{IEEE Transactions on Computer-Aided Design of Integrated Circuits and Systems {\bf 43}, 161–170}~(2024).

\bibitem{Bridgeman_2017}
Jacob~C Bridgeman and Christopher~T Chubb.
\newblock ``Hand-waving and interpretive dance: an introductory course on tensor networks''.
\newblock \href{https://dx.doi.org/10.1088/1751-8121/aa6dc3}{Journal of Physics A: Mathematical and Theoretical {\bf 50}, 223001}~(2017).

\bibitem{perezgarcia2007matrix}
D.~Perez-Garcia, F.~Verstraete, M.~M. Wolf, and J.~I. Cirac.
\newblock ``Matrix product state representations''~(2007).
\newblock  \href{http://arxiv.org/abs/quant-ph/0608197}{arXiv:quant-ph/0608197}.

\bibitem{Jarzyna_2013}
Marcin Jarzyna and Rafał Demkowicz-Dobrzański.
\newblock ``Matrix product states for quantum metrology''.
\newblock \href{https://dx.doi.org/10.1103/physrevlett.110.240405}{Physical Review Letters~{\bf 110}}~(2013).

\bibitem{ezzell2022quantum}
Nic Ezzell, Zoë Holmes, and Patrick~J. Coles.
\newblock ``The quantum low-rank approximation problem''~(2022).
\newblock  \href{http://arxiv.org/abs/2203.00811}{arXiv:2203.00811}.

\bibitem{chowdhury2020variational}
Anirban~N. Chowdhury, Guang~Hao Low, and Nathan Wiebe.
\newblock ``A variational quantum algorithm for preparing quantum gibbs states''~(2020).
\newblock  \href{http://arxiv.org/abs/2002.00055}{arXiv:2002.00055}.

\bibitem{Skinner_2019}
Brian Skinner, Jonathan Ruhman, and Adam Nahum.
\newblock ``Measurement-induced phase transitions in the dynamics of entanglement''.
\newblock \href{https://dx.doi.org/10.1103/physrevx.9.031009}{Physical Review X~{\bf 9}}~(2019).

\bibitem{Schumacher_1995}
Benjamin Schumacher.
\newblock ``Quantum coding''.
\newblock \href{https://dx.doi.org/10.1103/PhysRevA.51.2738}{Phys. Rev. A {\bf 51}, 2738--2747}~(1995).

\bibitem{Somoroff_2023}
Aaron Somoroff, Quentin Ficheux, Raymond~A. Mencia, Haonan Xiong, Roman Kuzmin, and Vladimir~E. Manucharyan.
\newblock ``Millisecond coherence in a superconducting qubit''.
\newblock \href{https://dx.doi.org/10.1103/PhysRevLett.130.267001}{Phys. Rev. Lett. {\bf 130}, 267001}~(2023).

\bibitem{dgn_24_low_depth_qsp}
Yulong Dong, Jonathan~A. Gross, and Murphy~Yuezhen Niu.
\newblock ``Optimal low-depth quantum signal-processing phase estimation''.
\newblock \href{https://dx.doi.org/10.1038/s41467-025-56724-x}{Nature Communications~{\bf 16}}~(2025).

\bibitem{laneve2023quantum}
Lorenzo Laneve.
\newblock ``Quantum signal processing over {SU(N)}''~(2024).
\newblock  \href{http://arxiv.org/abs/2311.03949}{arXiv:2311.03949}.

\bibitem{lu2024quantum}
Xi~Lu, Yuan Liu, and Hongwei Lin.
\newblock ``Quantum signal processing and quantum singular value transformation on {U(N)}''~(2024).
\newblock  \href{http://arxiv.org/abs/2408.01439}{arXiv:2408.01439}.

\bibitem{Elben_2022}
Andreas Elben, Steven~T. Flammia, Hsin-Yuan Huang, Richard Kueng, John Preskill, Benoît Vermersch, and Peter Zoller.
\newblock ``The randomized measurement toolbox''.
\newblock \href{https://dx.doi.org/10.1038/s42254-022-00535-2}{Nature Reviews Physics {\bf 5}, 9–24}~(2022).

\bibitem{rossi_m_qsp_22}
Zane~M. Rossi and Isaac~L. Chuang.
\newblock ``Multivariable quantum signal processing ({M-QSP}): prophecies of the two-headed oracle''.
\newblock \href{https://dx.doi.org/10.22331/q-2022-09-20-811}{Quantum {\bf 6}, 811}~(2022).

\bibitem{nemeth2023variants}
Bal{\'a}zs N{\'e}meth, Blanka K{\"o}v{\'e}r, Bogl{\'a}rka Kulcs{\'a}r, Roland~Botond Mikl{\'o}si, and Andr{\'a}s Gily{\'e}n.
\newblock ``On variants of multivariate quantum signal processing and their characterizations''~(2023).
\newblock  \href{http://arxiv.org/abs/2312.09072}{arXiv:2312.09072}.

\bibitem{gomes2024multivariable}
Niladri Gomes, Hokiat Lim, and Nathan Wiebe.
\newblock ``Multivariable qsp and bosonic quantum simulation using iterated quantum signal processing''~(2024).
\newblock  \href{http://arxiv.org/abs/2408.03254}{arXiv:2408.03254}.

\bibitem{alexeev2024quantum}
Yuri Alexeev, Maximilian Amsler, Marco~Antonio Barroca, Sanzio Bassini, Torey Battelle, Daan Camps, David Casanova, Young jai Choi, Frederic~T Chong, Charles Chung, et~al.
\newblock ``Quantum-centric supercomputing for materials science: A perspective on challenges and future directions''.
\newblock \href{https://dx.doi.org/https://doi.org/10.1016/j.future.2024.04.060}{Future Generation Computer Systems}~(2024).

\bibitem{liu2023bootstrap}
Yuan Liu, Oinam~R Meitei, Zachary~E Chin, Arkopal Dutt, Max Tao, Isaac~L Chuang, and Troy Van~Voorhis.
\newblock ``Bootstrap embedding on a quantum computer''.
\newblock \href{https://dx.doi.org/https://doi.org/10.1021/acs.jctc.3c00012}{Journal of Chemical Theory and Computation {\bf 19}, 2230--2247}~(2023).

\bibitem{bauer2020quantum}
Bela Bauer, Sergey Bravyi, Mario Motta, and Garnet Kin-Lic Chan.
\newblock ``Quantum algorithms for quantum chemistry and quantum materials science''.
\newblock \href{https://dx.doi.org/https://doi.org/10.1021/acs.chemrev.9b00829}{Chemical Reviews {\bf 120}, 12685--12717}~(2020).

\bibitem{Berry_2014}
Dominic~W. Berry, Andrew~M. Childs, Richard Cleve, Robin Kothari, and Rolando~D. Somma.
\newblock ``Exponential improvement in precision for simulating sparse hamiltonians''.
\newblock \href{https://dx.doi.org/http://dx.doi.org/10.1145/2591796.2591854}{Proceedings of the 46th Annual ACM Symposium on Theory of Computing}~(2014).

\bibitem{Berry_2015_Hamiltonian}
Dominic~W. Berry, Andrew~M. Childs, and Robin Kothari.
\newblock ``Hamiltonian simulation with nearly optimal dependence on all parameters''.
\newblock In 2015 IEEE 56th Annual Symposium on Foundations of Computer Science.
\newblock \href{https://dx.doi.org/10.1109/focs.2015.54}{Page 792–809}.
\newblock IEEE~(2015).

\bibitem{pyqsp}
Isaac Chuang, Andrew Tan, and John~M Martyn.
\newblock ``Py{QSP}: Python {Q}uantum {S}ignal {P}rocessing''.
\newblock \url{https://github.com/ichuang/pyqsp}~(2021).

\bibitem{rivlin2020chebyshev}
Theodore~J Rivlin.
\newblock ``Chebyshev polynomials''.
\newblock Courier Dover Publications. ~(2020).

\bibitem{boyd2001chebyshev}
John~P Boyd.
\newblock ``Chebyshev and fourier spectral methods''.
\newblock Courier Corporation. ~(2001).

\bibitem{rs_analytic_poly_02}
Q.~I. Rahman and G.~Schmiesser.
\newblock ``Analytic theory of polynomials''.
\newblock Clarendon Press, Oxford. ~(2002).

\bibitem{arfken2011mathematical}
George~B Arfken, Hans~J Weber, and Frank~E Harris.
\newblock ``Mathematical methods for physicists: a comprehensive guide''.
\newblock Academic press. ~(2011).

\bibitem{oliver1978note}
J~Oliver.
\newblock ``A note on the signs of truncated chebyshev polynomials''.
\newblock \href{https://dx.doi.org/https://doi.org/10.1007/BF01931700}{BIT Numerical Mathematics {\bf 18}, 233--235}~(1978).

\end{thebibliography}
\bibliographystyle{quantum}

\appendix

\section{Implementation of Arbitrary Polynomials with Linear-Combination-of-Unitaries and Generalized QSP}\label{app:Arbitrary_polynomials}
As we mentioned in Sec.~\ref{sec:Overview_QSP}, QSP generates polynomials that are restricted to obey the conditions of Eq.~\eqref{eq:qsp_conditions}, such as having definite parity. However, there exist techniques to expand this class of polynomials and implement arbitrary polynomials, provided they are bounded in magnitude. 

One method of achieving this is with linear-combination-of-unitaries (LCU) circuits~\cite{Berry_2014, Berry_2015_Simulating, Berry_2015_Hamiltonian}. In this context, an LCU circuit is composed of controlled QSP sequences and allows one to sum together polynomials. This can be used to construct a polynomial of indefinite parity by summing together its even and odd components, as we illustrate in Fig.~\ref{fig:LCU_circuit}. The LCU construction is sequential, implying that for a polynomial of degree $d$, the requisite query depth is $2d$. However, an LCU circuit rescales the sum by a constant, and thus requires amplitude amplification or additional measurements to accommodate for this rescaling. In the context of parallel QSP, where block encodings are multiplied together, this rescaling will accumulate exponentially in the number of threads and increase the measurement cost.

\begin{figure}[htbp]
    \centering
    \includegraphics[width=0.98\linewidth]{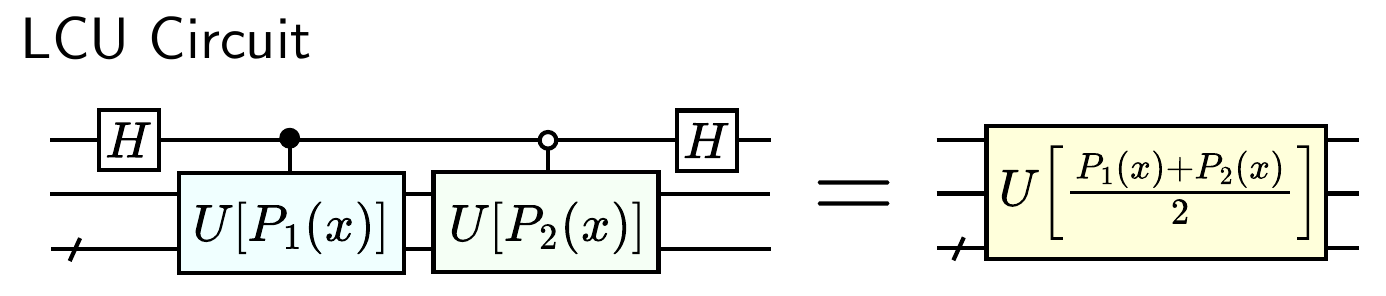}
    \caption{The linear-combination-of-unitaries (LCU) circuit that encodes a sum of two block encodings. Here $U [P_1(x)]$ and $ U [P_2(x)] $ block-encode polynomials $P_1(x)$ and $P_2(x)$, and are realized through QSP. Note how the sum of these polynomials is rescaled by a factor of $2$. Also observe that this construction is sequential, and thus doubles the query depth. }
    \label{fig:LCU_circuit}
\end{figure}

A more promising approach is the recently introduced construction of \emph{generalized QSP}~\cite{motlagh2023generalized}. Generalized QSP proposes a generalization of QSP sequence (Eq.~\eqref{eq:QSP_seqeunce}) that instead uses controlled block encoding:
\begin{equation}\label{eq:encodings_GQSP}
    \begin{bmatrix}
        U[A] & 0 \\
        0 & I 
    \end{bmatrix} , \qquad 
    \begin{bmatrix}
        I & 0 \\
        0 & U[A]^\dag 
    \end{bmatrix}. 
\end{equation}
By interspersing these with general SU(2) rotations (in contrast to $z$-rotations of Eq.~\eqref{eq:QSP_seqeunce}), Ref.~\cite{motlagh2023generalized} shows how the resulting sequence block-encodes a polynomials $P(U[A])$, restricted only by the constraint $\|P\|_{[-1,1]} \leq 1$. In particular, this construction requires $2d$ queries to the block encodings of Eq.~\eqref{eq:encodings_GQSP} to generate a degree-$d$ polynomial. This alleviates the parity constraint of ordinary QSP, while also avoiding the rescaling imposed by the LCU method. 

In the context of ordinary QSP, where we assume access to a block encoding, generalized QSP applies as follows. By a technique known as qubitization~\cite{Low_2019, Gily_n_2019, martyn2021grand}, which itself underlies QSP, there exists a privileged basis in which a block encoding of $A$ is equivalent to an $x$-rotation:
\begin{equation}
    U[A] = 
    \begin{bmatrix}
        A & i\sqrt{1-A^2} \\
        i\sqrt{1-A^2} & A
    \end{bmatrix} = e^{iX \arccos(A)} . 
\end{equation}
By then controlling on this block encoding and its inverse with an additional qubit, we can construct the inputs of Eq.~\eqref{eq:encodings_GQSP} needed for generalized QSP. Then, executing generalized QSP on this, we can encode a polynomial with coefficients $c_n$:
\begin{equation}
    P(U[A]) = \sum_{n=0}^d c_n U[A]^n = \sum_{n=0}^d c_n e^{iX n\arccos(A)} . 
\end{equation}
By projecting the remaining qubit into the $|0\rangle \langle 0 |$ element, we attain the polynomial \begin{equation}
\begin{aligned}
   \langle 0 | P(U[A]) | 0 \rangle & = \sum_{n=0}^d c_n \cos( n\arccos(A)) \\
   & = \sum_{n=0}^d c_n T_n(A) =: \tilde{P}(A)  
\end{aligned}
\end{equation}
where $T_n(x)$ is the order $n$ Chebyshev polynomial (see Appendix~\ref{app:bounds_partial_sum_polys} for review). As a linear combination of Chebyshev polynomials, this can represent an arbitrary degree-$d$ polynomial $\tilde{P}(A)$, provided it is bounded as $\| \tilde{P} \|_{[-1,1]} \leq 1$.

Together, these results are summarised as follows. A degree-$d$ polynomial of definite parity can be implemented directly through QSP with a query depth $d$. On the other hand, a degree-$d$ polynomial of indefinite parity can be implemented via generalized QSP at the expense of an additional control qubit, access to the inverse block-encoding unitary $U[A]^\dag$, and an increased query depth $2d$.

\section{Discussion of Code Implementation} \label{sec:dicussion_code_implementation}
The code for our implementation of parallel QSP can be found at Ref.~\cite{parallel_QSP}. With the goal of applying parallel QSP to the estimation of the property $w= \tr(P(\rho))$, this code takes as input a degree-$d$ polynomial $P(x)$ and an integer $k\geq 1$ corresponding to the number of threads over which to parallelize over. The code then executes the construction of Theorem~\ref{thm:parallel_qsp_prop_est_1}: it decomposes $P(x)$ into constituent polynomials as per Eq.~\eqref{eq:Poly_decomp}, assuming that $P(x)\geq 0$ is non-negative over the real axis. It then numerically finds the roots of $P_{\geq k}(x)$ to factorize it into $k$ factor polynomials according to Theorem~\ref{thm:parallel_qsp_poly_char}. This factorization is done arbitrarily and is not optimized to minimize the factorization constant $\mathcal{K}(P_{\geq k}) $; this optimization would be an interesting problem to study for future work.

For each resulting factor polynomials, we run a QSP phase finding algorithm. Using the $|+\rangle \langle + |$ block of a QSP sequence, this algorithm determines the QSP phases corresponding to the even/odd and real/imaginary components of the factor polynomials, such that the factor polynomials can be constructed via an LCU circuit. As the degree of each factor polynomial is reduced to $O(d/k)$, we can employ a simple phase finding algorithm mirroring that of Ref.~\cite{pyqsp}: we estimate the QSP phases by minimizing the sum of the squares error across $x\in [-1,1]$.

\section{Upper Bounds on the Magnitude of Constituent Polynomials} \label{app:bounds_partial_sum_polys}
Here we will prove bounds on properties of the constituent polynomials discussed in Sec.~\ref{sec:PQSP_for_Property_Estimation}. Consider a real degree-$d$ polynomial $P(x)$ that is bounded in magnitude by $1$ over $x \in [-1,1]$: 
\begin{equation}
   P(x) = \sum_{n=0}^d a_n x^n, \quad  \|P\|_{[-1,1]} \leq 1 .
\end{equation}
As in Sec.~\ref{sec:PQSP_for_Property_Estimation}, we split this polynomial into two constituent polynomials, $P_{< k} (x)$ and $P_{\geq k} (x)$, as
\begin{equation}
\begin{aligned}
    P(x) &= \sum_{n=0}^{k-1} a_n x^n + x^k \sum_{n=k}^d a_{n}x^{n-k} \\
    &=: P_{< k} (x) + x^k P_{\geq k} (x), 
\end{aligned}
\end{equation}
where
\begin{equation}\label{eq:partial_sum_polynomials_app}
\begin{aligned}
    & P_{< k} (x) := \sum_{n=0}^{k-1} a_n x^n, \quad P_{\geq k} (x) := \sum_{n=0}^{d-k} a_{n+k}x^{n}, 
\end{aligned}
\end{equation}
are the constituent polynomials of $P(x)$, of degree $k-1$ and $d-k$, respectively.

In order to estimate the trace $z = \tr( P(\rho))$ using parallel QSP, we estimate $w_{<k} = \tr( P_{< k}  (\rho))$ and $w_{\geq k} = \tr( P_{\geq k}  (\rho))$ separately. As shown in Sec.~\ref{sec:PQSP_for_Property_Estimation}, the measurement cost of obtaining these estimates depends on the magnitudes of the constituent polynomials $P_{< k} (x)$ and $P_{\geq k} (x)$. Although $\|P\|_{[-1,1]} \leq 1$, $P_{< k} (x)$ and $P_{\geq k} (x)$ are not necessarily bounded in magnitude by $1$. Below, we show that for any such polynomial $P(x)$, the constituent polynomials are bounded as 
\begin{equation}
    \begin{aligned}
        & \|P_{< k}\|_{[-1,1]} \leq O\Bigg( \frac{d^{k-1}}{(k-1)!} \Bigg) , \\
        & \|P_{\geq k}\|_{[-1,1]} \leq O\Bigg( \frac{d^k}{k!} \sqrt{\frac{d}{k}}\Bigg) . 
    \end{aligned}
\end{equation}
Therefore, $\|P_{< k}\|_{[-1,1]}, \|P_{\geq k} \|_{[-1,1]} \leq O({\rm poly}(d))$. This implies that the measurement cost of parallel QSP for property estimation is at worst polynomial in $d$.

In proving these bounds we will invoke the Chebyshev polynomials~\cite{rivlin2020chebyshev}. For context, recall that the $n$th Chebyshev polynomial $T_n(x)$ is a degree $n$ polynomial defined on $|x| \leq 1$ as
\begin{equation}
    T_n(x) = \cos(n \arccos(x)).
\end{equation}
It is well-established that $T_n(x)$ is a degree $n$ polynomial of fixed parity (either even or odd, depending on $n$) and bounded magnitude $|T_n(x)|_{[-1,1]} = 1$~\cite{boyd2001chebyshev, rivlin2020chebyshev}. They may be expressed as polynomials as
\begin{equation}
    T_n(x) = \sum_{j=0}^n t_{n,j} x^j,
\end{equation}
where $t_{n,j}$ are the corresponding coefficients.
It can also be shown that $T_n(x)$ can be re-expressed as
\begin{equation}\label{eq:Chebsyhev_reexpression}
    T_n(x) = \frac{1}{2} \Big( \big(x-\sqrt{x^2-1}\big)^n + \big(x+\sqrt{x^2-1} \big)^n \Big). 
\end{equation}

The Chebyshev polynomials provide a convenient basis for expanding functions on $x\in [-1,1]$. A function $F(x)$ can be expanded as
\begin{equation}
    F(x) = \frac{c_0}{2} + \sum_{n=1}^\infty c_n T_n(x) , 
\end{equation}
where
\begin{equation}
     c_n = \frac{2}{\pi} \int_{-1}^1 \frac{F(x) T_n(x)}{\sqrt{1-x^2}} dx 
\end{equation}
are the Chebyshev coefficients for all $n\geq 0$. These coefficients can be bounded as
\begin{equation}\label{eq:ChebyCoefBound}
    |c_n| \leq \frac{2}{\pi} \int_{-1}^1 \frac{\| F \|_{[-1,1]}}{\sqrt{1-x^2}} dx = \| F \|_{[-1,1]}. 
\end{equation}
Moreover, if $F(x)$ is degree-$d$ polynomial, then this series truncates at order $d$.

\subsection{Bound on the Magnitude of $P_{< k} (x)$}
To bound the magnitude of $P_{< k} (x)$, we begin by noting the result of the following theorems from Ref.~\cite{rs_analytic_poly_02}:
\begin{theorem}[Bound on Coefficients of Bounded Polynomials; Theorems~16.3.1 and~16.3.2 of Ref.~\cite{rs_analytic_poly_02}]\label{thm:CoefficientBound}
    Let $T_d(x) = \sum_{n=0}^d t_{d,n} x^n$ denote the Chebyshev polynomial of degree $d$. Let $P(x) = \sum_{n=0}^d a_n x^n$ be a degree-$d$ polynomial that is bounded at the Chebyshev nodes as
    \begin{equation}
        \begin{aligned}
            & \big|P\big(\cos(\tfrac{\pi \nu }{d})\big)\big| \leq 1 \quad  \text{for } \nu = 0, 1, ..., d , \\
            & \big|P\big(\cos \big(\tfrac{\pi \nu }{d-1}\big) \big)\big| \leq 1\  \text{for } \nu = 0, 1, ..., d-1 .
        \end{aligned}
    \end{equation}
    Then, for even $d$, the coefficients are bounded as
    \begin{equation}
        \begin{aligned}
            & |a_{2j}| \leq |t_{d,2j}| \\
            & |a_{2j-1} | \leq |t_{d-1,2j-1}|
        \end{aligned}
    \end{equation}
    for $j=0,1,..., d/2$. Analogously, for odd $d$, the coefficients are bounded as
    \begin{equation}
        \begin{aligned}
            & |a_{2j}| \leq |t_{d-1,2j}| \\
            & |a_{2j-1} | \leq |t_{d,2j-1}|
        \end{aligned}
    \end{equation}
    for $j=0,1,..., (d-1)/2$. Equality is achieved for $P(x) = T_d(x)$.
\end{theorem}

Applicability of this theorem demands that $P(x)$ be bounded by $1$ at the Chebyshev nodes; this condition is satisfied for a bounded polynomial $\|P\|_{[-1,1]} \leq 1$ as we consider here. Therefore, Theorem~\ref{thm:CoefficientBound} indicates that the coefficients of any such bounded polynomial are necessarily upper bounded in magnitude by the coefficients of Chebyshev polynomials. To employ this result in practice, it will be useful to also show that the coefficients of the Chebyshev polynomials scale as $|t_{d,n}| \leq O\big(\tfrac{d^n}{n!} \big)$:
\begin{lemma}[Bound on Coefficients of the Chebyshev Polynomials]\label{thm:ChebyshevCoefficientBound}
    The coefficients of the Chebyshev polynomials are bounded in magnitude as
    \begin{equation}
        |t_{d,n}| \leq \frac{(d+n)^n}{n!} = O\left( \frac{d^n}{n!} \right)
    \end{equation}
\end{lemma}
\begin{proof}
    Explicitly, the Chebyshev coefficients are~\cite{arfken2011mathematical}
    \begin{equation}\label{eq:ChebyPolyCoef}
        t_{d,n} = (-1)^{\frac{d-n}{2}} 2^{n-1} d \frac{\big( \frac{d+n}{2}-1\big)!}{\big( \frac{d-n}{2}\big) !\  n!} .
    \end{equation}
    This is upper bounded as 
    \begin{equation}
    \begin{aligned}
        |t_{d,n}| &= 2^{n-1} d \frac{\big( \frac{d+n}{2}-1\big)!}{\big( \frac{d-n}{2}\big) !\ n!} = \frac{2^{n-1} d}{\frac{d+n}{2}} {\frac{d+n}{2} \choose n} \\
        &\leq 2^n \frac{\big( \tfrac{d+n}{2} \big)^n} {n!} = \frac{(d+n)^n}{n!} = O\left( \frac{d^n}{n!} \right),
    \end{aligned}
    \end{equation}
    where we have used that $ {a \choose b} \leq \frac{a^b}{b!}$. 
\end{proof}

We can then merge Theorem~\ref{thm:CoefficientBound} and Lemma~\ref{thm:ChebyshevCoefficientBound} to bound $P_{< k} (x)$ as follows:
\begin{theorem}\label{thm:P^1_Bound}
    For any degree-$d$ polynomial $P(x)$ that is bounded as $\max_{x \in [-1,1]} |P(x)| \leq 1 $, its constituent polynomial $P_{< k} (x)$ (as defined in Eq.~\ref{eq:partial_sum_polynomials_app}) is necessarily bounded as
    \begin{equation}
        \| P_{< k} \|_{[-1,1]} \leq  O\left( \frac{d^{k-1}}{(k-1)!} \right) = {\rm poly}(d)  . 
    \end{equation}
\end{theorem}
\begin{proof}
    To bound the magnitude of $P_{< k} (x)$, we first note that the triangle inequality implies
    \begin{equation}
        \| P_{< k} \|_{[-1,1]}  = \max_{x \in [-1,1]} \left| \sum_{n=0}^{k-1} a_n x^n \right| \leq \sum_{n=0}^{k-1} |a_n| .
    \end{equation}
    In conjunction with Theorem~\ref{thm:CoefficientBound} and Lemma~\ref{thm:ChebyshevCoefficientBound}, we can bound this as
    \begin{equation}
        \sum_{n=0}^{k-1} |a_n| \leq \sum_{n=0}^{k-1} |t_{d,n}| \leq \sum_{n=0}^{k-1} O\left(\frac{d^{n}}{n!} \right) = O\left(\frac{d^{k-1}}{(k-1)!} \right). 
    \end{equation}
\end{proof}

Therefore, $P_{< k} (x)$ is necessarily bounded in magnitude by $O(d^{k-1}/(k-1)!)$. While a precise bound on $P_{< k} (x)$ depends on the coefficients of the polynomial $P(x)$ (and can grow slower than $O(d^{k-1}/(k-1)!)$), Theorem~\ref{thm:P^1_Bound} indicates that even in the worst case, the magnitude of $P_{< k} (x)$ only grows polynomially in $d$. We used this result in Sec.~\ref{sec:PQSP_for_Property_Estimation} to prove that the estimation of $w_{<k} = \tr( P_{< k}  (\rho))$ requires at most $\text{poly}(d)$ measurements.

\subsection{Bound on the Magnitude of $P_{\geq k} (x)$}
To bound $P_{\geq k} (x)$, it will not suffice to consider the sum of the magnitudes of the corresponding coefficients, as we did for $P_{< k} (x)$. In general, this sum can grow exponentially with $d$, which we aim to avoid. For example, for the Chebyshev polynomials, this sum is $\sum_n |t_{d,n}| = 2^{O(d)}$.

Instead, we will derive our result by first considering the constituent polynomials of the Chebyshev polynomials:
\begin{equation}
    T_n(x) =: T_{n;< k}(x) + x^k T_{n; \geq k}(x). 
\end{equation}
We will also consider the polynomial constructed from truncating the low order terms of a Chebyshev polynomial:
\begin{equation}
    \mathcal{T}_{n;k}(x) := \sum_{j=k}^n t_{n,j} x^j
\end{equation}
This is related to the constituent polynomial $T_{n; \geq k}(x)$ as
\begin{equation}\label{eq:truncated_Cheby_relation}
    T_{n; \geq k}(x) = \mathcal{T}_{n;k}(x)/x^k. 
\end{equation}
Because the Chebyshev polynomials have fixed parity, it is only relevant to consider $\mathcal{T}_{n;k}$ for $n$ and $k$ of the same parity (i.e., both either even or odd). Ref.~\cite{oliver1978note} studied these truncated Chebyshev polynomials and showed they have definite sign:
\begin{theorem}[Main Result of Ref.~\cite{oliver1978note}]\label{thm:TruncatedChebySign}
    The truncated Chebyshev polynomials $\mathcal{T}_{n;k}(x)$ have definite sign over $x \in [0,1]$:
    \begin{equation}
    \begin{aligned}
        & (-1)^l \mathcal{T}_{n;n-2l}(x) \geq 0 \text{ for } x \in [0,1],
    \end{aligned}
    \end{equation}
    for all $l=0,1,...,\lfloor \frac{n}{2} \rfloor -1 $. 
\end{theorem}
This implies that the sign of $\mathcal{T}_{n;n-2l}(x)$ over $x \in [0,1]$ is equal to the sign of the coefficient $t_{n,n-2l}$. In the region $x \in [-1,0]$, the sign of is determined by the parity of $\mathcal{T}_{n;n-2l}(x)$, or equivalently the parity of $n$. We can use this result to prove the following bound on $T_{n; \geq k}(x)$.

\begin{corollary}[Maximum of $T_{n; \geq k}(x)$]\label{thm:Cheby_2MagnitudeBound}
For $n$ and $k$ of the same parity (i.e., both even or odd), the maximum magnitude of the constituent polynomial $T_{n; \geq k}(x)$ over $x \in [-1,1]$ occurs at $x=0$ and takes the value:
\begin{equation}
    \max_{x \in [-1,1]}|T_{n; \geq k}(x)| = |T_{n; \geq k}(0)| = |t_{n,k}| = O\left(\frac{n^k}{k!} \right) .
\end{equation}
\end{corollary}
\begin{proof}
    The proof will proceed by induction on increasing $n$. First, note that for $n$ and $k$ of the same parity, the partial sum polynomial $T_{n; \geq k}(x)$ consists of only even powers, and hence is an even function. According to Theorem~\ref{thm:TruncatedChebySign} and Eq.~\eqref{eq:truncated_Cheby_relation}, this function is of constant sign over $-1 \leq x \leq 1$, and this sign is $\text{sign}(t_{n,k}) = (-1)^{\frac{n-k}{2}}$. 

    Moving to the proof by induction, we will make the inductive hypothesis that the maximum magnitudes of the constituent polynomials are achieved at $x=0$: 
    \begin{equation}
        \forall n' \leq n, \ |T_{n';\geq k}(0)| \geq |T_{n';\geq k}(x)|, 
    \end{equation}
    for all $k=0,2,..., n'$ if $n'$ is even, or $k=1, 3, ..., n'$ if $n'$ is odd. 
    
    Let us show that the base cases $n'=0$ and $n'=1$ are satisfied. For $n'=0$, we have $|T_{0; \geq 0}(x)| = 1 \leq |T_{0;\geq 0}| = 1$ is satisfied trivially. For $n'=1$, we similarly have $|T_{1;\geq 1}(x)| = |t_{1,1}| \leq |T_{1; \geq 0}(0)| = |t_{1,1}|$. It is then straightforward to consider larger values of $n'$, such as $n'=2$ in which case we have
    \begin{equation}
        \begin{aligned}
            & |T_{2; \geq 0}(x)| \leq 1 = |T_{2; \geq 0}(0)| ,  \text{ and} \\
            & |T_{2; \geq 2}(x)| = |t_{2,2}| \leq |T_{2 ;\geq 2}(0)| = |t_{2,2}|. 
        \end{aligned}
    \end{equation}
    at $k=0$ and $k=2$, respectively. 

    Then, supposing that the inductive hypothesis is true up to $n'=n$, we can show that it is also true for $n'=n+1$. To prove this, note that the Chebyshev polynomials obey the recursion relation 
    \begin{equation}\label{eq:Cheby_recurrence}
        T_{n+1}(x) = 2x T_n(x) - T_{n-1}(x).
    \end{equation}
    This implies that, for $k$ of the same parity as $n+1$, the constituent polynomials obey
    \begin{equation}\label{eq:truncated_recurrence}
        T_{n+1; \geq k}(x) = 2 T_{n ; \geq k-1}(x) - T_{n-1; \geq k}(x)
    \end{equation}
    for $k\leq n-1$ (which ensures that the polynomial $T_{n-1; \geq k}(x)$ exists). As per Theorem~\ref{thm:TruncatedChebySign} and Eq.~\eqref{eq:truncated_Cheby_relation}, $T_{n; \geq k-1}(x)$ and $T_{n-1; \geq k}(x)$ have constant and opposite sign over $-1 \leq x \leq 1$. This implies that
    \begin{equation}
    \begin{aligned}
        |T_{n+1; \geq k}(x)| &\leq 2 |T_{n; \geq k-1}(x)| + |T_{n-1; \geq k}(x)| \\
        &\leq 2 |T_{n; \geq k-1}(0)| + |T_{n-1; \geq k}(0)| \\
        &= |2T_{n; \geq k-1}(0) - T_{n-1; \geq k}(0)| \\
        &= |T_{n+1; \geq k}(0)|,
    \end{aligned}
    \end{equation}
    where this first inequality is an application of the triangle inequality, the second inequality is the inductive hypothesis, and the last equalities follow from the constant sign of the constituent polynomials. This holds true for $k=0,..., n-1$ as per Eq.~\eqref{eq:truncated_recurrence}. We can prove the remaining cases $k=n$ and $k=n+1$ as follows. For $k=n$, the recurrence relation of Eq.~\eqref{eq:Cheby_recurrence} corresponds to 
    \begin{equation}
        T_{n+1; \geq n}(x) = 2T_{n; \geq n-1}(x),
    \end{equation}
    such that 
    \begin{equation}
    \begin{aligned}
        |T_{n+1; \geq n}(x)| &= 2 |T_{n; \geq n-1}(x)| \\
        &\leq  2 |T_{n; n-1}(0)| = |T_{n+1; \geq n}(0)|.
    \end{aligned}
    \end{equation}
    For $k=n+1$, the hypothesized inequality is trivially satisfied because the constituent polynomial is a constant: $|T_{n+1; \geq n+1}(0)| = |t_{n+1,n+1}| \geq |T_{n+1; \geq n+1}(0)| = |t_{n+1,n+1}|$. 
    
    Therefore, by induction on increasing $n$, this completes the proof that the maximum of the constituent polynomial $T_{n; \geq k}(x)$ occurs at $x=0$. Moreover, the value at $x=0$ is $|T_{n; \geq k}(0)| = |t_{n,k}| = O\big( \tfrac{n^k}{k!} \big) $. 
\end{proof}

We can use Corollary~\ref{thm:Cheby_2MagnitudeBound} in conjunction with the Chebyshev decomposition of a polynomial to prove the following bound on an arbitrary constituent polynomial:
\begin{theorem}[Bound on constituent polynomial $P_{\geq k} (x)$]\label{thm:P^2_Bound}
For a polynomial $P(x)$ that is bounded as $\|P \|_{[-1,1]} \leq 1$, its partial sum polynomial is necessarily bounded as
\begin{equation}
    \|P_{\geq k} \|_{[-1,1]} \leq O\Bigg( \frac{d^k}{k!} \sqrt{\frac{d}{k}} \Bigg) = O\big( {\rm poly}(d) \big) .
\end{equation}
\end{theorem}
\begin{proof}
    To prove this result, first decompose $P(x)$ into the basis of Chebyshev polynomials:
    \begin{equation}
        P(x) = \sum_{n=0}^d c_n T_n(x)
    \end{equation}
    where $c_n$ are the Chebyshev coefficients of $P(x)$. The Chebyshev polynomials obey the orthogonality relation $\frac{2-\delta_{n,0}}{\pi} \int_{-1}^1 \frac{T_n(x) T_m(x) }{\sqrt{1-x^2}} dx = \delta_{nm}$, such that the coefficients $c_n$ are given by
    \begin{equation}\label{eq:Cheby_c_n}
        c_n = \frac{2-\delta_{n,0}}{\pi} \int_{-1}^1 dx \frac{P(x) T_n(x) }{\sqrt{1-x^2}} .
    \end{equation}
    In this basis, the constituent polynomial can be expressed as
    \begin{equation}
        P_{\geq k} (x) = \sum_{n=k}^d c_n T_{n; \geq k}(x).
    \end{equation}
    
    Next, we can employ the Cauchy-Schwarz inequality to show that
    \begin{equation}\label{eq:CS_bound}
        \begin{aligned}
            |P^{2,k}(x)| &= \Big| \sum_{n=k}^d c_n T_{n; \geq k}(x) \Big| \\
            &\leq \sqrt{\sum_{n=k}^d |c_n|^2} \times \sqrt{\sum_{n=k}^d \big|T_{n; \geq k}(x) \big|^2 }
        \end{aligned}
    \end{equation}
    Applying Parseval's theorem with the inner product of the Chebyshev polynomials, the 2-norm of $c_n$ is upper bounded as 
    \begin{equation}
    \begin{aligned}
        &\frac{2}{\pi} \int_{-1}^1 dx \frac{|P(x)|^2}{\sqrt{1-x^2}} = \frac{|c_0|^2}{2} + \sum_{n=1}^d |c_n|^2 \\
        & \qquad \qquad \qquad \qquad \leq \frac{2}{\pi} \int_{-1}^1 dx \frac{1}{\sqrt{1-x^2}} = 1,
    \end{aligned}
    \end{equation}
    and therefore $\sum_{n=k}^d |c_n|^2 \leq 2 = O(1)$. On the other hand, the 2-norm of the Chebyshev polynomials can be upper bounded by invoking Corollary~\ref{thm:Cheby_2MagnitudeBound}:
    \begin{equation}
        \begin{aligned}
            &\max_{x\in [-1,1]} \sum_{n=k}^d \big|T_{n; \geq k}(x) \big|^2 \leq 
            \sum_{n=k}^d |t_{n,k}|^2 \leq 
            \sum_{n=k}^d O\Bigg( \frac{n^{2k}}{(k!)^2}\Bigg) \\
            & = O\Bigg( \frac{1}{(k!)^2} \cdot \frac{d^{2k+1}}{2k+1} \Bigg) = O\Bigg( \Big( \frac{d^k}{k!} \Big)^2 \frac{d}{k} \Bigg) , 
        \end{aligned}
    \end{equation}
    where we have used that $\sum_{n=0}^d n^{2k} = \frac{d^{2k+1}}{2k+1} + O(d^{2k})$. Inputting these upper bound into Eq.~\eqref{eq:CS_bound}, we obtain 
    \begin{equation}
        |P_{\geq k} (x)| \leq O\Bigg( \frac{d^k}{k!} \sqrt{\frac{d}{k}} \Bigg) ,
    \end{equation}
    for $x \in [-1,1]$. This completes the proof of the stated result. 
    
    As an aside, we suspect this bound could be sharpened to $\| P_{\geq k} \|_{[-1,1]} \leq O\big( \tfrac{d^k}{k!} \big)$, which is saturated by the Chebyshev polynomials. 
\end{proof}

\section{Augmenting Trace Estimation with Importance Sampling}\label{app:Importance_Sampling}
Importance sampling can be utilized to expand the class of functions whose traces can be estimated. To demonstrate this for functions of a density matrix, suppose that we have the ability to estimate the trace of a set of basis functions $\{B_j(\rho)\}_{j=1}^d$, by using QSP or other techniques. For example, one could have $B_j(\rho) = \rho^j$ be monomials, or even $B_j(\rho) = T_j(\rho)$ be the Chebyshev polynomials. In practice, the trace of a basis function is approximated by repeatedly measuring an estimator $\hat{\mathcal{B}}_j$, whose expectation value is the desired trace: 
\begin{equation}
    \mathop{\mathbb{E}} [\hat{\mathcal{B}}_j] = \tr(B_j(\rho))
\end{equation}
For example, in the QSP test (Sec.~\ref{sec:QSP_Trace_Est}), the estimator is $\hat{\mathcal{B}} = m \in \{0,1\} $, where $m$ is the measurement of the block-encoding qubit.

Given this ability, one can expand the class of functions whose traces can be estimated by incorporating importance sampling. Consider a function $f(\rho)$ expanded in the basis $\{ B_j(\rho) \}_{j=1}^d$ as 
\begin{equation}
    f(\rho) = \sum_{j=1}^d c_j B_j(\rho), 
\end{equation}
with complex coefficients $c_j$. In order to estimate the trace $\tr(f(\rho))$, first define a probability distribution $p(j) = |c_j|/\|c\|_1$, where $\|c\|_1 = \sum_j |c_j|$ is the 1-norm of the coefficients. Then, by sampling an integer $j \sim p(j)$ and evaluating the corresponding estimator $\hat{\mathcal{B}}_j$, one can construct the following quantity whose expectation value yields the desired trace:
\begin{equation}
\begin{aligned}
    &\mathop{\mathbb{E}}_{j \sim p} \left[ \frac{c_j}{|c_j|} \hat{\mathcal{B}}_j \right] = 
    \sum_{j=1}^d p_j \frac{c_j}{|c_j|} \tr(B_j(\rho)) =
    \frac{\tr(f(\rho))}{\|c\|_1}. 
\end{aligned}
\end{equation}
Due to the rescaling by $\|c\|_1$, estimating $\tr(f(\rho))$ to additive error $\epsilon$ requires $O\big(\|c\|_1^2/\epsilon^2\big)$ measurements.

Ref.~\cite{tosta2023randomized} uses this importance sampling procedure to estimate a large class of traces and expectation values, provided the ability to generate Chebyshev polynomials $B_j(\rho)=T_j(\rho)$ with QSP. A similar sampling procedure was also used in Ref.~\cite{Quek_2024} to estimate the trace of functions of a density matrix, given only the ability to estimate the trace of the monomials $B_j(\rho) = \rho^j$. In both references, because the measurement cost depends quadratically on the 1-norm of the coefficients, the authors note that this method is best suited for well-behaved functions whose 1-norm is not prohibitively large (i.e., scales only polynomially in $d$ rather than exponentially). 

In our work, we use importance sampling to extend parallel QSP from the limited scope of Theorem~\ref{thm:parallel_qsp_poly_char} to a larger class of property estimation problems according to Theorem~\ref{thm:parallel_qsp_prop_est_3}. We achieve this by decomposing a trace $\tr(P(\rho))$ into a linear combination of terms that are each amenable to parallel QSP, and applying importance sampling to this sum. We discuss the proof of this theorem in the following section.

\section{Proofs for Parallel QSP for Property Estimation}\label{app:ParallelQSPPropertyEstimation}
Here we prove Theorem~\ref{thm:parallel_qsp_prop_est_3}, which provides a general scheme for estimating a property $w = \tr(P(\rho))$ with parallel QSP. As we mentioned in the main text, we achieve this by decomposing $P(x)$ into a linear combination of polynomials that are each amenable to parallel QSP. We can then apply importance sampling to this linear combination to extract $w$.

The basis that we choose to decompose into is the basis of Chebyshev polynomials. In this basis, we can present and prove the theorem:
\begin{theorem}[Parallel QSP for Property Estimation: Definite Parity] 
Let $P(x)$ be a real polynomial of degree $d$ and definite parity, that is bounded as $\| P \|_{[-1,1]} \leq 1$. By invoking parallel QSP across $k$ threads, where $k$ has the same parity as $d$, we can estimate
\begin{equation}
    w = \tr(P(\rho))
\end{equation}
with a circuit of width $O(k)$ and query depth at most $\lfloor \frac{d-k}{2k} \rfloor + k-1 = O(d/k+k)$. The number of measurements required to resolve $w$ to additive error $\epsilon$ is 
\begin{equation}
\begin{aligned}
    &O\Bigg( \frac{\|P_{< k} \|_{[-1,1]}^2}{\epsilon^2} + \frac{\|P_{< k} \|_{[-1,1]}^2 d^4 \big( 1+\sqrt{2} \big)^{4k} }{k^2 \epsilon^2 } \Bigg) \\
    &= O\Bigg( \frac{\|P_{< k} \|_{[-1,1]}^2 + \|P_{\geq k} \|_{[-1,1]}^2 d^4 2^{O(k)}}{\epsilon^2} \Bigg).
\end{aligned}
\end{equation}
\end{theorem}

\begin{proof}
    Suppose $P(x)$ is of definite parity, and that $k$ has the same parity. Then, by decomposing $P(x) = P_{< k} (x) + x^k P_{\geq k} (x)$, $P_{< k} (x)$ has the same parity as $P(x)$, whereas $P_{\geq k} (x)$ is necessarily even. As in the proof of Theorem~\ref{thm:parallel_qsp_prop_est_1}, we seek to estimate $w_{<k}$ and $w_{\geq k}$ to error $\epsilon/2$ each, such that their sum approximates $w$ to error $\epsilon$. 

    First, we can estimate $w_{<k}$ with standard QSP. Because $P_{< k} (x)$ is real and of definite parity, the requisite query depth is $k-1$ and the requisite number of measurements is $O \left( \| P_{< k}  \|_{[-1,1]}^2 / \epsilon^2 \right)$.

    Next, because $P_{\geq k} (x)$ is even, we can expand it in the basis of even Chebyshev polynomials:
    \begin{equation}\label{eq:P^2_ChebyDecomp}
    \begin{aligned}
        P_{\geq k} (x) &= \sum_{j=0}^{\frac{d-k}{2}} c_{2j} T_{2j}(x) \\
        &= \sum_{a=0}^{\lfloor \frac{d-k}{2k} \rfloor } \sum_{b=0}^{k-1} c_{a\cdot 2k + 2b} T_{a\cdot 2k + 2b}(x).
    \end{aligned}
    \end{equation}
    where we have judiciously recast this as a sum over Chebyshev polynomials with orders expressed as multiples of $2k$, i.e. order $2j = a \cdot 2k+b$ for $a,b \geq 0$. 
    Going forward, our approach will be to use properties of the Chebsyhev polynomials to recast $P_{\geq k} (x)$ as a linear combination of polynomials that are each amenable to parallel QSP. Then, one can estimate the corresponding trace of the terms in this linear combination to furnish an approximation to $w_{\geq k}$. 
    
    In more detail, we will use the following properties of the Chebyshev polynomials~\cite{rivlin2020chebyshev}:
    \begin{equation}
        \begin{aligned}
            &T_{mn}(x) = T_m(T_n(x)) \\
            & T_{m+n}(x) = 2T_m(x) T_n(x) - T_{|m-n|}(x).
        \end{aligned}
    \end{equation}
    These properties enable, for instance, a Chebyshev polynomial of degree $4n$ to be recast as
    \begin{equation}
    T_{4n}(x) = T_4(T_{n}(x)) = 8 T_{n}(x)^4 - 8 T_{n}(x)^2 + 1. 
    \end{equation}
    The polynomials comprising this linear combination ($T_{n}(x)^4, T_{n}(x)^2$, and $1$) are all positive definite and factorize trivially as products of Chebyshev polynomials of degree at most $n$. Thus, each term in this expression is amenable to parallel QSP, and the degree is reduced by a factor of $4$. This is the strategy we will use going forward, but with the degree reduced by a factor of $2k$ instead of $4$. 

    Applying this strategy to Eq.~\eqref{eq:P^2_ChebyDecomp}, we begin by by writing
    \begin{equation}
        T_{a\cdot 2k +2b}(x) = 2 T_{a\cdot 2k}(x) T_{2b}(x) - T_{a\cdot 2k  - 2b}(x)
    \end{equation}
    for $b > 0$ and $a\cdot 2k-2b \geq 0 $. Inserting this expression into Eq.~\eqref{eq:P^2_ChebyDecomp}, each term $ -T_{a\cdot 2k  - 2b}(x)$ can be re-included into the sum by modifying the corresponding coefficient $c_{a\cdot 2k  - 2b}$. Starting with the highest degree $a=\lfloor (d-k)/2k \rfloor$, this allows us to write 
    \begin{equation}
    \begin{aligned}
        &\sum_{b=1}^{k-1} c_{a\cdot 2k+2b} T_{a\cdot 2k+2b}(x) = \\
        &\sum_{b=1}^{k-1} c_{a\cdot 2k+2b} \big( 2 T_{a\cdot 2k}(x) T_{2b}(x) - T_{(a-1)\cdot2k+2(k-b)}(x) \big).
    \end{aligned}
    \end{equation}
    By re-including these terms into the sum of Eq.~\eqref{eq:P^2_ChebyDecomp}, this effectively changes the coefficients to (for $b > 0$)
    \begin{equation}\label{eq:ChebyCoefMapping}
    \begin{aligned}
        &c_{a\cdot 2k+2b} \mapsto 2 c_{ak+b} \\
        &c_{(a-1)\cdot 2k+2(k-b)} \mapsto c_{(a-1)\cdot 2k+2(k-b)} - c_{a\cdot 2k+2b}.
    \end{aligned}
    \end{equation}
    We now decrement $a$, and recurse this procedure, updating the coefficients appropriately. Note that the second mapping in Eq.~\eqref{eq:ChebyCoefMapping} is equivalent to $c_{ak+b} \mapsto c_{ak+b} - c_{(a+1)k+(k-b)} $, which becomes the recursion
    \begin{equation}
    \begin{aligned}
        c_{a\cdot 2k+2b} \ \mapsto \   &c_{a\cdot 2k+2b} - c_{(a+1)\cdot 2k+2(k-b)} \\
        +  &c_{(a+2)\cdot 2k+2b} -  c_{(a+3) \cdot 2k+2(k-b)} ... . 
    \end{aligned}
    \end{equation}
    Unfolding this recursion for each term in Eq.~\eqref{eq:P^2_ChebyDecomp}, we find that the new coefficients are
    \begin{equation}
    \begin{aligned}
        \tilde{c}_{a\cdot 2k + 2b} = 
        \begin{cases}
            c_{a\cdot 2k} & b=0 \\
            c_{2b} & a=0 \\ 
            2\sum_{j=0}^{\lfloor (d-k)/2k \rfloor - a} (-1)^j c_{(a+j)\cdot 2k+2\alpha_j} & a, b \geq 1, \\ 
        \end{cases}
    \end{aligned}
    \end{equation}
    where 
    \begin{equation}
        \alpha_j = \begin{cases}
                b & j \text { even} \\
                k-b & j \text { odd},
            \end{cases}
    \end{equation}
    such that we may rewrite $P_{\geq k} (x)$ as 
    \begin{equation}\label{eq:P^2_ChebyDecomp2}
    \begin{aligned}
        &P_{\geq k} (x) = \sum_{a=0}^{\lfloor \frac{d-k}{2k} \rfloor} \sum_{b=0}^{k-1} \tilde{c}_{a\cdot 2k+2b} T_{a\cdot 2k}(x) T_{2b}(x).
    \end{aligned}
    \end{equation}
    Because $|c_{2j}| \leq \| P_{\geq k} \|_{[-1,1]}$ as per Eq.~\eqref{eq:ChebyCoefBound}, the magnitude of these coefficients is 
    \begin{equation}\label{eq:c_tilde_bound}
        | \tilde{c}_{a\cdot 2k+2b} | \leq 2 \| P_{\geq k} \|_{[-1,1]} \Big( \Big\lfloor \frac{d-k}{2k} \Big\rfloor - a \Big) . 
    \end{equation}

    Next, let us denote the Chebyshev polynomials of even degree as 
    \begin{equation}
        T_{2n}(x) = \sum_{j=0}^{n} t_{2n,2j} x^{2j}.
    \end{equation}
    Accordingly, we can express the product $T_{a\cdot 2k}(x) T_{2b}(x)$ in Eq.~\eqref{eq:P^2_ChebyDecomp2} as     
    \begin{equation}
    \begin{aligned}
        T_{a\cdot 2k}(x) T_{2b}(x) &= T_{2k}(T_a(x)) T_{2}(T_b(x)) \\
        &= \sum_{j=0}^{k} \sum_{l=0}^1 t_{2k,2j} t_{2,2l} T_a(x)^{2j} T_b(x)^{2l}. 
    \end{aligned}
    \end{equation}
    such that 
    \begin{equation}
    \begin{aligned}
        &P_{\geq k} (x) \\
         \quad &=\sum_{a=0}^{\lfloor \frac{d-k}{2k} \rfloor} \sum_{b=0}^{k-1} \sum_{j=0}^{k} \sum_{l=0}^1 \tilde{c}_{a\cdot 2k+2b} t_{2k,2j} t_{2,2l} T_a(x)^{2j} T_b(x)^{2l} \\
         & =: \sum_{a=0}^{\lfloor \frac{d-k}{2k} \rfloor} \sum_{b=0}^{k-1} \sum_{j=0}^{k} \sum_{l=0}^1 C^{(k)}_{abjl} T_a(x)^{2j} T_b(x)^{2l}, 
    \end{aligned}
    \end{equation}
    where we have defined the coefficients $C^{(k)}_{abjl} = \tilde{c}_{a\cdot 2k+2b} t_{2k,2j} t_{2,2l}$. 

    This allows us to write $w_{\geq k}$ as
    \begin{equation}\label{eq:w_2_decomp}
    \begin{aligned}
        w_{\geq k} &= \tr(\rho^k P_{\geq k} (\rho)) \\
        &=: \sum_{a=0}^{\lfloor \frac{d-k}{2k} \rfloor} \sum_{b=0}^{k-1} \sum_{j=0}^{k} \sum_{l=0}^1 C_{abjl}^{(k)} \tr(\rho^k \big|T_a(\rho) ^{j} T_b(\rho)^{l} \big| ^2).
    \end{aligned}
    \end{equation}
    This is now re-expressed as a linear combination of terms that are each amenable to parallel QSP. That is, the trace $\tr(\rho^k \big|T_a(\rho) ^{j} T_b(\rho)^{l} \big| ^2)$ obeys the conditions of Theorem~\ref{thm:parallel_qsp_poly_char}: the polynomial $R(x) = \big(T_a(x) ^{j} T_b(x)^{l} \big) ^2$ is real and non-negative, and trivially factorizes into $k$ factor polynomials as
    \begin{equation}
        R(x) = \prod_{s=1}^k |\mathcal{R}_s(x)|^2
    \end{equation}
    where
    \begin{equation}
        \mathcal{R}_s(x) = 
        \begin{cases}
            T_a(x) T_b(x) & s=1,\ \text{for } j\geq 1, l=1,  \\
            T_b(x) & s=1,\ \text{for } j=0, l=1,  \\
            T_a(x) & s\geq 1, \ j>l, \\
            1 & s \geq j,l .
        \end{cases}
    \end{equation}
    These factor polynomials are all real-valued, of definite parity, and have degree at most $a+b \leq \lfloor \frac{d-k}{2k} \rfloor + k-1$. Therefore, they can each be directly implemented through QSP, with query depth at most $\lfloor \frac{d-k}{2k} \rfloor + k-1 = O(d/k+k)$, and achieve a factorization constant $\mathcal{K}=1$. 

    Using this strategy, we can estimate $w_{\geq k}$ by estimating the terms in the linear combination of Eq.~\eqref{eq:w_2_decomp} with parallel QSP. A particularly efficient way to do this is to use importance sampling according to Appendix~\ref{app:Importance_Sampling}. That is, define a probability distribution $p(a,b,j,l) = |C^{(k)}_{a,b,j,l}| / \| C^{(k)} \|_1$. Then, by sampling $a,b,j,l \sim p(a,b,j,l) $ and evaluating the corresponding estimator of $\tr(\rho^k \big|T_a(\rho) ^{j} T_b(\rho)^{l} \big| ^2)$ (i.e. the measurement of the parallel QSP circuit associated with this polynomial), we can construct a quantity whose expectation value converges to $w_{\geq k}$. The associated cost of estimating $w_{\geq k}$ to additive error $\epsilon$ is $O(\| C^{(k)} \|_1^2/ \epsilon^2)$. 
    
    In order to bound thus cost, we can upper bound the 1-norm $\| C^{(k)} \|_1$ by invoking the identity 
    \begin{equation}
        \sum_{j=0}^n |t_{2n,2j}| = \frac{1}{2}(1+\sqrt{2})^{2n} + \frac{1}{2}(1-\sqrt{2})^{2n} = O((1+\sqrt{2})^{2n}) .
    \end{equation}
    This follows from the fact that $t_{2n,2j} = (-1)^j (-1)^n |t_{2n,2j}| $ (see Eq.~\eqref{eq:ChebyPolyCoef}), such that the 1-norm of the Chebsyhev polynomial coefficients can be expressed as 
    \begin{equation}
        (-1)^n T_n(i) = (-1)^n \sum_{j=0}^n t_{2n,2j} (i)^{2j} = \sum_{j=0}^n |t_{2n,2j}|.
    \end{equation}
    Using Eq.~\eqref{eq:Chebsyhev_reexpression}, this evaluates to $(-1)^n T_n(i) = \frac{1}{2}(1+\sqrt{2})^{2n} + \frac{1}{2}(1-\sqrt{2})^{2n}$. Therefore, using this identity, we can upper bound $\| C^{(k)} \|_1$ as
    \begin{equation}
    \begin{aligned}
        \| C^{(k)} \|_1 &= \sum_{a=0}^{\lfloor \frac{d-k}{2k} \rfloor} \sum_{b=0}^{k-1} \sum_{j=0}^{k} \sum_{l=0}^1 | \tilde{c}_{a\cdot 2k+2b} t_{2k,2j} t_{2,2l} | \\
        &= \sum_{a=0}^{\lfloor \frac{d-k}{2k} \rfloor} \sum_{b=0}^{k-1} | \tilde{c}_{a\cdot 2k+2b}| \sum_{j=0}^{k} |t_{2k,2j}| \cdot  \sum_{l=0}^1 |t_{2,2l}| \\
        &= \| \tilde{c} \|_1 \cdot O\big( (1+\sqrt{2})^{2k} \big).
    \end{aligned}
    \end{equation}
    Using Eq.~\eqref{eq:c_tilde_bound}, we can upper bound as $\| \tilde{c} \|_1$
    \begin{equation}
    \begin{aligned}
         \| \tilde{c} \|_1 &= \sum_{a=0}^{\lfloor \frac{d-k}{2k} \rfloor} \sum_{b=0}^{k-1} | \tilde{c}_{a\cdot 2k+2b}| \\
        & \leq \sum_{a=0}^{\lfloor \frac{d-k}{2k} \rfloor} \sum_{b=0}^{k-1}  2 \| P_{\geq k} \|_{[-1,1]} \Big( \Big\lfloor \frac{d-k}{2k} \Big\rfloor - a \Big) \\
        &= 2 k \| P_{\geq k} \|_{[-1,1]} \sum_{a=0}^{\lfloor \frac{d-k}{2k} \rfloor} \Big( \Big\lfloor \frac{d-k}{2k} \Big\rfloor - a \Big) \\
        &= O\Big( \| P_{\geq k} \|_{[-1,1]} \frac{d^2}{k}  \Big).
    \end{aligned}
    \end{equation}
   And therefore we obtain a measurement cost
    \begin{equation}
    \begin{aligned}
        O\Big( \frac{\| C^{(k)} \|_1^2 }{\epsilon^2} \Big) &= O\Bigg( \frac{ \| P_{\geq k} \|_{[-1,1]}^2 d^4 \big( 1+\sqrt{2} \big)^{4k} }{k^2 \epsilon^2 } \Bigg) \\
        &= O\Bigg( \frac{\| P_{\geq k} \|_{[-1,1]}^2 d^4 2^{O(k)}}{\epsilon^2} \Bigg).
    \end{aligned}
    \end{equation}

    Lastly, we suspect that this bound could be tightened by improving the bound of Eq.~\eqref{eq:c_tilde_bound}, which is rather loose. 

\end{proof}

\end{document}